\documentclass[pdflatex,sn-mathphys-num]{sn-jnl}

\usepackage{graphicx}%
\usepackage{multirow}%
\usepackage{amsmath,amssymb,amsfonts}%
\usepackage{amsthm}%
\usepackage{mathrsfs}%
\usepackage[title]{appendix}%
\usepackage{xcolor}%
\usepackage{textcomp}%
\usepackage{manyfoot}%
\usepackage{booktabs}%
\usepackage{algorithm}%
\usepackage{algorithmicx}%
\usepackage{algpseudocode}%
\usepackage{listings}%
\usepackage{amssymb,latexsym,amsmath}     
\usepackage{epsf}
\usepackage{array}
\usepackage{amssymb,amscd,comment}
\usepackage{float}
\usepackage{comment, verbatim}
\usepackage{geometry}
\usepackage{relsize}
\usepackage{calc}
\usepackage[mathscr]{eucal}
\usepackage{hhline}
\usepackage{tikz}
\usetikzlibrary{positioning}
\usepackage{enumitem,mathrsfs}
\setlist[enumerate]{leftmargin=.5in}
\setlist[itemize]{leftmargin=.5in}
\usepackage{multirow}
\usepackage{tabularx}
\usepackage{soul, color}
\usepackage{amssymb,amstext,amsgen,amsfonts,mathrsfs,verbatim,url,hyperref,mathdots, mathtools}
\usepackage{pgfplots}
\pgfplotsset{compat=1.18}
\usepackage{mathrsfs}
\usepackage{arydshln}

\ifpdf
  \DeclareGraphicsExtensions{.eps,.pdf,.png,.jpg}
\else
  \DeclareGraphicsExtensions{.eps}
\fi

\newlist{inparaenum}{enumerate}{2}
\setlist[inparaenum]{nosep}
\setlist[inparaenum,1]{label=\bfseries\alph*.}

\def\comment#1{}
\newcommand{\C}{\mathbb{C}}
\newcommand{\R}{\mathbb{R}}

\newcommand{\e}{\mathbf{e}}

\usetikzlibrary{arrows}

\def\invddots{\mathinner{\mskip1mu\raise1pt\vbox{\kern7pt\hbox{.}}\mskip2mu
		\raise4pt\hbox{.}\mskip2mu\raise7pt\hbox{.}\mskip1mu}}

\theoremstyle{thmstyleone}%
\newtheorem{theorem}{Theorem}
\newtheorem{proposition}[theorem]{Proposition}%
\newtheorem{lemma}[theorem]{Lemma}%
\newtheorem{corollary}[theorem]{Corollary}%

\theoremstyle{thmstyletwo}%
\newtheorem{remark}{Remark}%

\theoremstyle{thmstylethree}%
\newtheorem{definition}{Definition}%

\begin{document}




\author[1]{\fnm{Mia} \sur{Escobar}}\email{mescob3@uw.edu}
\equalcont{These authors contributed equally to this work.}

\author[2]{\fnm{Valentin} \sur{Garcia}}\email{valentin\_garcia@brown.edu}
\equalcont{These authors contributed equally to this work.}

\author[3]{\fnm{Anastasiia} \sur{Minenkova}}\email{minenkova@hartford.edu}
\equalcont{These authors contributed equally to this work.}

\affil[1]{\orgdiv{Department of Mathematics}, \orgname{University of Washington - Tacoma}, \orgaddress{\street{1900 Commerce Street}, \city{Tacoma}, \postcode{98402}, \state{WA}, \country{USA}}}

\affil[2]{\orgdiv{Department of Mathematics}, \orgname{Brown University}, \orgaddress{\street{ 1 Prospect St}, \city{Providence}, \postcode{02912}, \state{RI}, \country{USA}}}

\affil[3]{\orgdiv{Department of Mathematics}, \orgname{University of Hartford}, \orgaddress{\street{200 Bloomfield Ave.}, \city{West Hartford}, \postcode{06117}, \state{CT}, \country{USA}}}

\title[Spectral Fusion and ESE]{Spectral Fusion for Identifying Early State Exclusion in Symmetric Quantum Spin Chains}

\abstract{
Perfect state transfer (PST) in one-dimensional quantum spin chains provides a natural setting in which quantum information transport can be analyzed using spectral methods. In the single-excitation subspace, the dynamics of a chain with nearest-neighbor interactions are governed by a Jacobi matrix, allowing questions of state transfer to be formulated in terms of eigenvalue distributions and symmetry properties of eigenvectors. In this work, we investigate the phenomenon of \emph{early state exclusion} (ESE), whereby the overlap of the time-evolved state with the initial state vanishes at a time strictly earlier than the first occurrence of perfect state transfer. Building on earlier constructions of Hamiltonians exhibiting PST with and without ESE, we provide explicit Jacobi matrix realizations for arbitrary odd-length chains and establish general conditions under which ESE occurs or does not occur. We propose the process of \emph{spectral fusion} (SF) to build infinite families of such Hamiltonians. These results broaden the known class of spin chains displaying early state exclusion and further clarify the role of spectral structure of the Hamiltonians.}

\maketitle






\section{Introduction}

The problem of quantum state transfer in one-dimensional spin systems has attracted sustained interest due to its relevance in quantum information processing and its rich mathematical structure. In particular, \emph{perfect state transfer} (PST) refers to the phenomenon whereby a quantum state initially localized at one site of a spin chain is reproduced exactly at another site after a finite evolution time. This capability is essential for the coherent transmission of quantum information between spatially separated components of a quantum device.

From a physical perspective, quantum communication can be realized either by directly transmitting quantum states or by distributing entanglement and performing teleportation protocols. While optical fibers are the natural medium for long-distance quantum communication, short-range communication between nearby quantum processors motivates the study of alternative physical channels, including engineered spin chains and ion-trap systems. These systems admit effective Hamiltonian descriptions that lend themselves naturally to analytical and spectral investigation.

The study of state transfer in spin chains was initiated by Bose~\cite{B03}, who considered a one-dimensional chain of $N$ qubits with open boundary conditions governed by a time-independent Hamiltonian. A key structural insight was later provided by Kay~\cite{Kay10}, who showed that the Hamiltonian describing a chain of $N+1$ qubits with nearest-neighbor interactions can be represented, in the single-excitation subspace, by a real symmetric tridiagonal matrix, i.e., a Jacobi matrix. This observation established a direct link between quantum dynamics and the spectral theory of Jacobi operators.

Subsequent developments (see, for example,~\cite{Kay10,VZh12,Kall} and references therein) revealed deep connections between perfect state transfer, graph symmetries, orthogonal polynomials, and inverse spectral problems. In this framework, questions of PST and its variants, such as \emph{pretty good state transfer}, may be formulated in terms of  conditions on eigenvalues and symmetry constraints on eigenvectors. As a result, the study of quantum state transfer has become a fertile meeting ground for quantum information theory, graph theory, and numerical linear algebra.

In the present paper, we investigate the phenomenon of \emph{early state exclusion} (ESE). In~\cite{ESE}, the authors constructed families of one-dimensional Hamiltonians that not only realize perfect state transfer between the end vertices of a weighted path, but also exhibit the property that the overlap of the time-evolved state with the initial state vanishes at some time strictly before the first PST time. When this occurs, the system is said to exhibit early state exclusion. It was shown that infinite families of Hamiltonians admit PST both with and without ESE, depending on spectral properties and the parity of the chain length. 

Our goal is to extend this analysis to the case of an odd number of qubits and to provide explicit constructions of Hamiltonians for which early state exclusion occurs. We focus on chains with nearest-neighbor interactions (see Figure~\ref{fig:nearestneighbor}), whose dynamics in the single-excitation subspace are governed by a Jacobi matrix
\begin{equation}\label{eq:J}
    J=\begin{bmatrix}
		a_0 & b_0    &  &  \\
		b_0 & a_1    & \ddots &  \\
		& \ddots & \ddots & b_{N-1} \\
		&        & b_{N-1}& a_{N}     \\
	\end{bmatrix},
\end{equation}	
where $a_k \in \R$ and $b_k > 0$. Denote by $\sigma(J)$ the spectrum of $J$. The time evolution of the system is then given by $e^{-iJt}$. Thus, the spectral properties of $J$ determine the dynamical behavior of the system and play a central role in the analysis of both perfect state transfer and early state exclusion.

\begin{figure}[ht]
    \begin{center}
    \begin{tikzpicture}[roundnode/.style={circle, draw=gray! 150!, fill=gray!10, very thick, minimum size=12mm},scale=1.3]
        \node[roundnode] (1) {$0$};
        \node[roundnode] (2) [right=of 1]{$1$};
        \node[roundnode] (3) [right=of 2]{$2$};
        \node[roundnode] (N) [right=of 3]{\small$N-1$};
        \node[roundnode] (N+1) [right=of N]{\small$N$};
         \path (1) edge [-stealth, gray! 150, very thick,out=30,in=150] node[above] {$b_0$} (2);
          \path (2) edge [-stealth, gray! 150, very thick,out=-150,in=-30] (1);
              \path (2) edge [-stealth, gray! 150, very thick,out=30,in=150] node[above] {$b_1$} (3);
          \path (3) edge [-stealth, gray! 150, very thick,out=-150,in=-30] (2);
             \path (3) edge [-stealth, dotted,gray! 150, very thick,out=30,in=150] (N);
          \path (N) edge [-stealth, dotted, gray! 150, very thick,out=-150,in=-30] (3);
                   \path (N) edge [-stealth, gray! 150, very thick,out=30,in=150] node[above] {$b_{N-1}$} (N+1);
          \path (N+1) edge [-stealth, gray! 150, very thick,out=-150,in=-30] (N);
        \path (1) edge [-stealth, gray! 150, very thick,out=80,in=100,looseness=7] node[above] {$a_0$} (1);
            \path (2) edge [-stealth, gray! 150, very thick,out=80,in=100,looseness=7] node[above] {$a_1$} (2);
                \path (3) edge [-stealth, gray! 150, very thick,out=80,in=100,looseness=7] node[above] {$a_2$} (3);
                    \path (N) edge [-stealth, gray! 150, very thick,out=80,in=100,looseness=7] node[above] {$a_{N-1}$} (N);
                        \path (N+1) edge [-stealth, gray! 150, very thick,out=80,in=100,looseness=7] node[above] {$a_N$} (N+1);
        \coordinate (A) at (1,1);
    \end{tikzpicture}
    \end{center}
    \caption{A quantum spin-chain model consisting of $N + 1$ qubits with nearest-neighbor coupling and environmental interactions corresponding to Hamiltonian expressed by $J$.}
    \label{fig:nearestneighbor}
\end{figure}
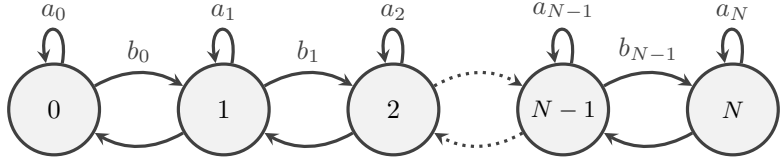

In this paper, we describe the spectral properties of $J$ that result in either ESE or no ESE, and introduce a method to generate infinite families of both through \textit{spectral fusion} (SF). Section~\ref{Section 2} covers the preliminary information necessary to set up of the problem, along with previously proven examples of odd and even-ordered spin chains. Section~\ref{Section 3} then establishes a formula for the amplitude of odd and even-ordered systems, alongside the key properties of these amplitudes. Finally, Section~\ref{Section 4} introduces SF as the inverse operation of \textit{spectral surgery}, allowing us to predict how a system will change when eigenvalues are added to a chain, and, importantly, how to conserve properties such as ESE, thereby proving the existence of infinite families of odd-ordered quantum spin chains with and without ESE.

\section{Preliminaries}
\label{Section 2}
Using the canonical basis $\{\e_k\}_{k=0}^N$ of $\C^{N+1}$, where $\e_k$ represents the excited state of the $(k+1)$th qubit and, in particular, $\e_0 = [1,0,\dots,0]^\top$, we study the time evolution of the initial state $\e_0$ under the unitary time evolution operator $e^{-iJt}$. Under suitable conditions on $J$, there exists a time $T>0$ such that the probability of finding the excitation at the last qubit is equal to one. This observation motivates the following definition.
\begin{definition}
\emph{A Jacobi matrix $J$ realizes \textbf{perfect state transfer} (PST) between the end-vertices of the weighted path in Figure~\ref{fig:nearestneighbor} at time $T > 0$ if
\begin{equation*}
\label{eq:pst}
e^{-iJT}\e_0 = e^{i\phi}\e_N \quad \text{or} \quad |\langle e^{-iJt}{\e}_0,{\e_N}\rangle_{\mathbb{C}^{N+1}}| = 1
\end{equation*}
for some phase $\phi \in \R$.}
\end{definition}

\begin{remark}
{From here on, when we state that a matrix $J$ realizes PST, this refers exclusively to PST between the terminal vertices of its associated weighted path graph. }
\end{remark}

Kay~\cite{Kay10} established that a necessary and sufficient condition for $J$ to realize PST is that $J$ must be persymmetric and its ordered eigenvalues $\mu_0 < \mu_1 < \cdots < \mu_N$ satisfy the relation
\begin{equation}
\label{eq:spectrumPST}
\mu_{k + 1} - \mu_k = \frac{(2n_k + 1)\pi}{T} \quad \text{for every $0 \leq k \leq N - 1$,}
\end{equation}
where each $n_k$ is a nonnegative integer.

In~\cite{ESE}, the authors introduced the following phenomenon.
\begin{definition}
Let $J$ be a Jacobi matrix that has the earliest PST between the end-vertices of the weighted path occurring at time $T_0 > 0$. If there is a time $\tau$ such that $0 < \tau < T_0$ and
\[
\langle e^{-iJ\tau}{\bf e}_0,{\bf e_0}\rangle_{\mathbb{C}^{N+1}} = 0,
\]
then $J$ has an \textit{Early State Exclusion} (ESE) at time $\tau$.
\end{definition}

\begin{figure}[H]
    \centering
    \includegraphics[width=0.5\linewidth]{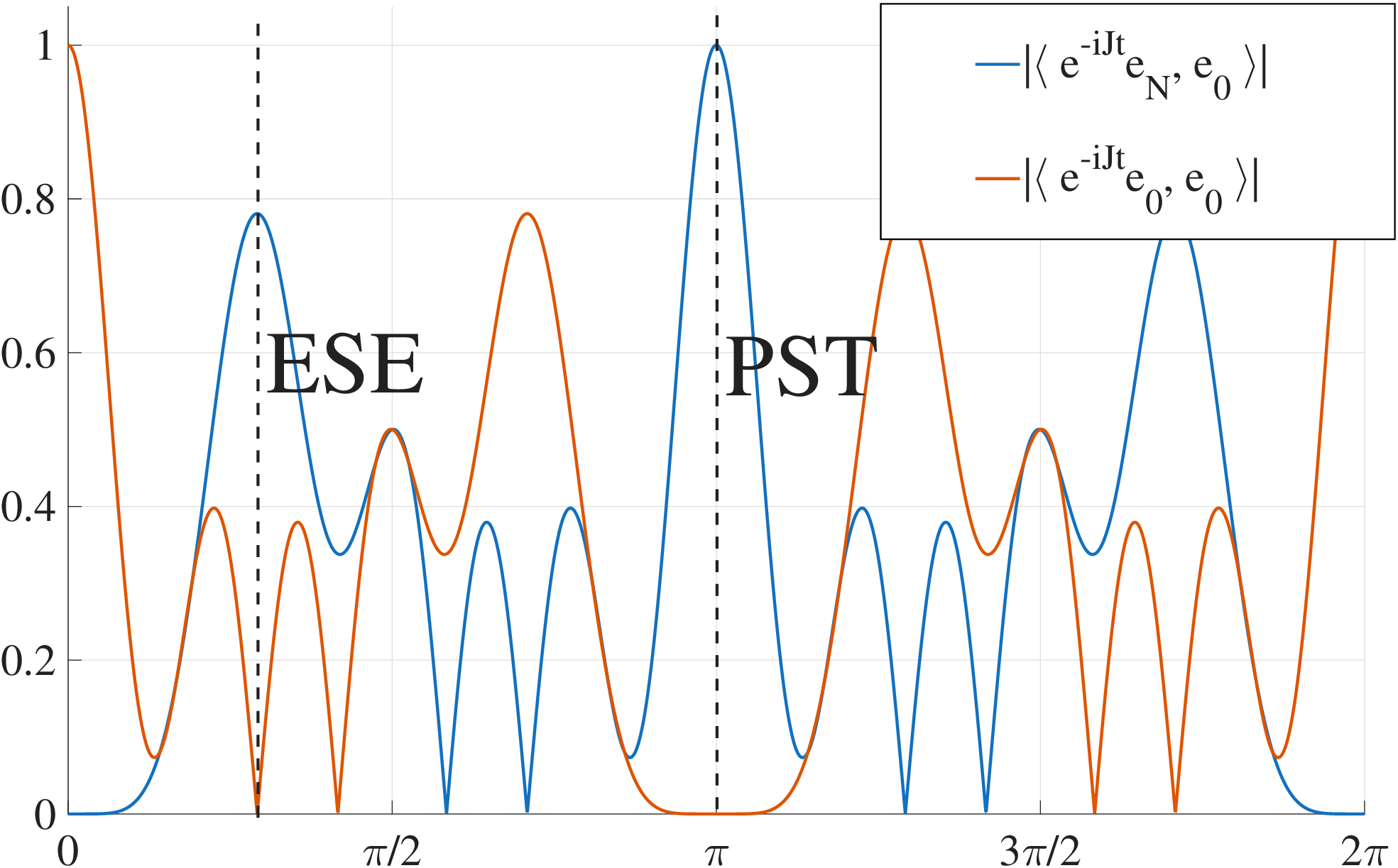}
    \caption{Examples of ESE occurrences for $J$ with the reduced spectrum $\{0,\pm3,\pm8,\pm9\}$.}
    \label{fig:389}
\end{figure}

The ESE phenomenon is not necessarily unique and can happen multiple times (see Figure~\ref{fig:389}).
In~\cite{ESE}, the authors discussed that ESE would allow one to safely re‑inject quantum information into the system before the original transfer completes, potentially increasing transmission efficiency.
The authors first analyzed small systems to establish limitations. They proved that ESE is impossible for systems with either $2$ or $3$ qubits, showing that any time at which the initial state vanished necessarily coincided with PST time  itself. These results rely on spectral constraints imposed by PST, such as eigenvalue spacing and persymmetry of the Jacobi matrix. The analysis clarifies that not every Hamiltonian satisfying the PST conditions can exhibit ESE, and that system size and structure play a crucial role.
The main contribution of the paper is the construction of Jacobi matrices that exhibit both PST and ESE, starting with an explicit $4\times4$ example and then extending to an infinite family for all even system sizes. Using properties of Krawtchouk and Chebyshev polynomials, the authors systematically build Hamiltonians where the overlap with the initial state vanishes at an intermediate time, yet PST still occurs later.

In related works~\cite{EM25,EscobarGarcia2025}, the authors constructed Hamiltonians for chains of five and seven qubits with symmetric spectra and investigated the conditions under which ESE  occurs. They constructed examples when the smallest positive eigenvalue divides all other eigenvalues, in which case ESE does not occur, and examples when ESE appears and this divisibility condition fails. In the case of 5 qubit chains, this divisibility condition was a criterion for the Hamiltonian to exhibit no ESE. Based on these examples, it was conjectured that this spectral property characterizes the presence or absence of ESE in systems with symmetric spectra. In this work, we demonstrate that this is not the case in general for larger chains. In particular, we provide explicit counterexamples showing that the divisibility of eigenvalues by the smallest positive eigenvalue is not sufficient for the ``no ESE'' Hamiltonian. This indicates that ESE is governed by more subtle spectral and structural features of the Hamiltonian than previously suggested. 

\section{Jacobi Matrices and the Amplitudes}
\label{Section 3}
Let $J$ be a persymmetric Jacobi matrix of order $N + 1$ with spectrum $\sigma(J)$. We say that $\sigma(J)$ is \emph{symmetric }
if $-\lambda \in \sigma(J)$ whenever $\lambda  \in \sigma(J)$. Since the eigenvalues of $J$ are distinct~\cite{BG}, $0 \in \sigma(J)$ whenever $N + 1$ is odd. Hence, when studying persymmetric Jacobi matrices with symmetric spectra, one must account for the parity of $N + 1$. We therefore adopt the following convention:
\begin{itemize}
    \item If $N + 1 = 2n + 1$ is odd, we write $\sigma(J) = \{\lambda_k\}_{k = -n}^n$ with $\lambda_{-k} = -\lambda_k$.
    \item If $N + 1 = 2n$ is even, we write $\sigma(J) = \{\pm \lambda_k\}_{k = 1}^n$.
\end{itemize}
Furthermore, we assume that $0 < \lambda_1 < \cdots < \lambda_n$ in both cases.

\begin{remark}
If $J$ experiences PST and $\sigma(J)$ is \emph{centered} about $\mu \in \R$ (i.e., $\mu + \lambda  \in \sigma(J)$ if and only if $\mu - \lambda \in \sigma(J)$), the resulting amplitude formulas in Proposition~\ref{Prop: amplitudes} differ only by a factor of $e^{-i\mu t}$. Because a translation of the spectrum preserves both PST and the number of instances in which ESE occurs (see~\cite{EM25}), it suffices to study the case where $\mu = 0$.
\end{remark}

We introduce a family of polynomials related to $J$ via a three-term recurrence relation. Specifically, we define the sequence $\{p_k\}_{k = 1}^N$ by: 
	\begin{equation}\label{threeP}
		p_{k+1}(x) = p_{k}(x)(x-a_{k}) - b^2_{k-1}p_{k-1}(x) \quad \text{for every $1 \leq k \leq N$,}
	\end{equation}
where each $a_k$ and $b_k$ are given as in \eqref{eq:J}, $p_0 \equiv 1$, $p_{-1} \equiv 0$, and $b_{-1} \equiv 1$. Note that 
    \[p_k(x)=\det(xI-J_k) \quad \text{for every $1 \leq k \leq N+1$, }\]
where $J_k$ is the $k\times k$ truncation of $J$ obtained by deleting its last $N+1-k$ rows and columns (see~\cite{DMS}).

\begin{lemma}
\label{Lemma b's}
Let $J$ be a persymmetric Jacobi matrix of order $N + 1$ with a symmetric spectrum. 
Then
\[
\frac{1}{b_0 \cdots b_{N - 1}}
=
\begin{cases}
    \displaystyle \sum_{j = 0}^{n} \prod_{\substack{i = 0 \\ i \neq j}}^n \frac{1}{|\lambda_j^2 - \lambda_i^2|} &\text{if $N + 1 = 2n + 1$ is odd} \\
    \displaystyle \sum_{j = 1}^{n} \frac{1}{\lambda_j} \prod_{\substack{i = 1 \\ i \neq j}}^n \frac{1}{|\lambda_j^2 - \lambda_i^2|} &\text{if $N + 1 = 2n$ is even,}
\end{cases}
\]
where each $b_k$ is given as in \eqref{eq:J}.
\end{lemma}
\begin{proof}
Assume that $N + 1 = 2n + 1$ is odd. Then the characteristic polynomial of $J$ is given by
\begin{equation}
\label{eq:p_{N + 1} odd}
    p_{N + 1}(x) = x 
    \prod_{i = 1}^n \big(x - 
    \lambda_i\big)\big(x +
    \lambda_i\big)  = x 
    \prod_{i = 1}^{n}(x^2 - 
    \lambda_i^2).
\end{equation}
Differentiating the polynomial above yields
\[
    p_{N + 1}'(x) =  \prod_{i = 1}^{n}(x^2 - 
    \lambda_i^2) + 2x
    ^2\sum_{j = 1}^{n} \prod_{\substack{i = 1 \\ i \neq j}}^n (x^2 - 
    \lambda_i^2).
\]
Evaluating $p_{N + 1}'$ at each $\lambda_j$ gives
\begin{equation}
\label{eq:p_{N + 1}' odd}
|p_{N + 1}'(
\lambda_j)| = 
    \begin{cases}
            \displaystyle \prod_{i = 1}^n \lambda_i^2 &\text{if $j = 0$} \\
            \displaystyle 2\lambda_j^2 \prod_{\substack{i = 1 \\ i \neq j}}^n |\lambda_j^2 - \lambda_i^2| &\text{if $j = 1, \dots, n$.}
    \end{cases}
\end{equation}
Hence, by Theorem 3 in~\cite{H74},
    \[
    \frac{1}{b_0 \cdots b_{N - 1}} = \sum_{j = -n}^n \frac{1}{|p_{N + 1}'(
    \lambda_j)|}=\sum_{j = 0}^{n} \prod_{\substack{i = 0 \\ i \neq j}}^n \frac{1}{|\lambda_j^2 - \lambda_i^2|}.
    \]
\par Now assume that $N + 1 = 2n$ is even. Noting that the characteristic polynomial of $J$ is
\begin{equation}
\label{eq:p_{N + 1} even}
    p_{N + 1}(x) = \prod_{i = 1}^n (x^2 
    - \lambda_i^2)
\end{equation}
allows us to calculate
\begin{equation*}
\label{eq:p_{N + 1}' even}
    |p_{N + 1}'(
    \lambda_j)| = 2\lambda_j \prod_{\substack{i = 1 \\ i \neq j}}^n |\lambda_j^2 - \lambda_i^2|.
\end{equation*}
Again, by Theorem 3 in~\cite{H74},
\[
    \frac{1}{b_0 \cdots b_{N - 1}} = 2\sum_{j = 1}^n \frac{1}{|p_{N + 1}'(
    \lambda_j)|} = \sum_{j = 1}^{n} \frac{1}{\lambda_j} \prod_{\substack{i = 1 \\ i \neq j}}^n \frac{1}{|\lambda_j^2 - \lambda_i^2|}.
\]
\end{proof}

Our next result provides a concrete description of the ESE amplitude function solely through the eigenvalues of $J$ for a symmetric spectrum.

\begin{proposition}[Amplitude Formulas]
\label{Prop: amplitudes}
Suppose that $J$ is a persymmetric Jacobi matrix of order $N + 1$ with a symmetric spectrum. 
Then 
\begin{equation*}
\label{eq:Amplitude Formulas}
\langle e^{-iJt} \e_0, \e_0 \rangle_{\mathbb{C}^{N+1}} =
\begin{cases}
    \displaystyle 
    \sum_{k = 0}^{n} c_k\cos{\lambda_k t} &\text{if $N + 1 = 2n + 1$ is odd} \\
    \displaystyle 
    \sum_{k = 1}^{n} \widetilde{c}_k\cos{\lambda_k t} &\text{if $N + 1 = 2n$ is even,}
\end{cases}
\end{equation*}
where
\begin{equation}
\label{eq:c_k}
c_k = \prod_{\substack{i = 0 \\ i \neq k}}^n \frac{1}{|\lambda_k^2 - \lambda_i^2|} \Bigg/\sum_{j = 0}^{n} \prod_{\substack{i = 0 \\ i \neq j}}^n \frac{1}{|\lambda_j^2 - \lambda_i^2|} \quad \text{for every $0 \leq k \leq n$}
\end{equation}
and
\begin{equation}
\label{eq:widetilde{c}_k}
\widetilde{c}_k = \frac{1}{\lambda_k} \prod_{\substack{i = 1 \\ i \neq k}}^n \frac{1}{|\lambda_k^2 - \lambda_i^2|} \Bigg/\sum_{j = 1}^{n} \frac{1}{\lambda_j} \prod_{\substack{i = 1 \\ i \neq j}}^n \frac{1}{|\lambda_j^2 - \lambda_i^2|} \quad \text{for every $1 \leq k \leq n$.}
\end{equation}
\end{proposition}

\begin{proof}
We again distinguish between the parity of $N + 1$. First, assume that $N + 1 = 2n + 1$ is odd. Define $\{P_j\}_{k = 1}^{N}$ to be the collection of normalized polynomials $P_j = p_j/(b_0 \cdots b_{j - 1})$. 

\par Consider the set of vectors $\{\mathbf{u}_k\}_{k = -n}^n \subseteq \R^{N + 1}$ given by
\begin{equation*}
    \mathbf{u}_k = \big(P_0(
    \lambda_k), P_1(
    \lambda_k), \dots, P_N(
    \lambda_k)\big)^\top \quad \text{for every $-n \leq k \leq n$,}
\end{equation*}
where $P_0 \equiv 1$. Then $\mathbf{u}_k$ is an eigenvector of $J$ with eigenvalue $
\lambda_k$~\cite{BG}. Expanding $\e_0$ under this eigenbasis yields
\[
\e_0 = \sum_{k = -n}^{n} \frac{\langle \e_0, \mathbf{u}_k \rangle_{\mathbb{C}^{N+1}} \mathbf{u}_k}{\|\mathbf{u}_k\|^2} = \sum_{k = -n}^{n} \frac{\mathbf{u}_k}{\|\mathbf{u}_k\|^2}.
\]
Hence
\begin{equation*}
    \langle e^{-iJt}\e_0, \e_0 \rangle_{\mathbb{C}^{N+1}} = \sum_{k = -n}^{n} \frac{\langle e^{-iJt}\mathbf{u}_k, \e_0 \rangle_{\mathbb{C}^{N+1}}}{\|\mathbf{u}_k\|^2} = \sum_{k = -n}^{n} \frac{\langle e^{-i
    \lambda_kt}\mathbf{u}_k, \e_0 \rangle_{\mathbb{C}^{N+1}}}{\|\mathbf{u}_k\|^2} = 
    \sum_{k = -n}^{n} \frac{e^{-i\lambda_kt}}{\|\mathbf{u}_k\|^2}.
\end{equation*}
Since $\{P_j\}_{k = 1}^{N}$ alternates in parity
, $\|\mathbf{u}_k\| = \|\mathbf{u}_{-k}\|$. Thus
\begin{align*}
    \langle e^{-iJt}\e_0, \e_0 \rangle_{\mathbb{C}^{N+1}} &= \frac{ 1
    }{\|\mathbf{u}_0\|^2} +  
    \sum_{k = 1}^{n} \frac{e^{i\lambda_kt} + e^{-i\lambda_kt}}{\|\mathbf{u}_k\|^2} =  
    \sum_{k = 0}^{n} c_k\cos{\lambda_kt},
\end{align*}
where $c_0 = 1/\|\mathbf{u}_0\|^2$ and $c_k = 2/\|\mathbf{u}_k\|^2$ for every $k \geq 1$.

To compute $\|\mathbf{u}_k\|^2$, we apply Theorem 3 in~\cite{H74} and Lemma~\ref{Lemma b's} to get
\begin{equation*}
\|\mathbf{u}_k\|^2 = \frac{|p_{N + 1}'(
\lambda_k)|}{b_0 \cdots b_{N - 1}} = |p_{N + 1}'(
\lambda_k)| \cdot \sum_{j = 0}^{n} \prod_{\substack{i = 0 \\ i \neq j}}^n \frac{1}{|\lambda_j^2 - \lambda_i^2|}.
\end{equation*}
Substituting \eqref{eq:p_{N + 1}' odd} into the equation above yields the coefficients in \eqref{eq:c_k}. 

\par Now suppose that $N + 1 = 2n$ is even. Consider the same normalized polynomials $\{P_j\}_{k = 1}^{N}$ as before, and define a set of vectors $\{\mathbf{v}_{\pm k}\}_{k = 1}^n \subseteq \R^{N + 1}$ by
\begin{equation*}
\mathbf{v}_{\pm k} = \big(P_0(
\pm \lambda_k), P_1(
\pm \lambda_k), \dots, P_N(
\pm \lambda_k)\big)^\top \quad \text{for every $1 \leq k \leq n$.}
\end{equation*}
An almost identical analysis as before shows us that
\[
\langle e^{-iJt}\e_0, \e_0 \rangle_{\mathbb{C}^{N+1}} = 
\sum_{k = 1}^{n} \frac{e^{i\lambda_kt} + e^{-i\lambda_kt}}{\|\mathbf{v}_k\|^2}= \sum_{k = 1}^{n} \widetilde{c}_k \cos{\lambda_kt}
\]
with $\widetilde{c}_k = 2/\|\mathbf{v}_k\|^2$ for every $k \geq 1$. Once again using Theorem 3 in~\cite{H74} and Lemma~\ref{Lemma b's} with \eqref{eq:p_{N + 1}' odd} gives us that
\begin{align*}
\|\mathbf{v}_k\|^2 &= \frac{|p_{N + 1}'(\mu + \lambda_k)|}{b_0 \cdots b_{N - 1}} = 2\lambda_k \prod_{\substack{i = 1 \\ i \neq k}}^n |\lambda_k^2 - \lambda_i^2| \cdot \sum_{j = 1}^{n} \frac{1}{\lambda_j} \prod_{\substack{i = 1 \\ i \neq j}}^n \frac{1}{|\lambda_j^2 - \lambda_i^2|},
\end{align*}
This matches our claimed values of $\widetilde{c}_k = 2/\|\mathbf{v}_k\|^2$ in \eqref{eq:widetilde{c}_k}, completing our proof.
\end{proof}

\begin{remark}
\label{Remark 1.3}
{Using similar techniques as above, one can show that the relevant inner product value for PST is given by
\[
(-1)^n\langle e^{-iJt}\e_0, \e_N \rangle_{\mathbb{C}^{N+1}} = \begin{cases} \displaystyle 
\sum_{k = 0}^n (-1)^k c_k\cos{\lambda_kt} &\text{if $N + 1 = 2n + 1$ is odd} \\
\displaystyle 
\sum_{k = 1}^n (-1)^k \widetilde{c}_k\cos{\lambda_kt} &\text{if $N + 1 = 2n$ is even,}
\end{cases}
\]
provided that $J$ is a persymmetric Jacobi matrix with a symmetric spectrum. 
}
\end{remark}
Fix a persymmetric Jacobi matrix $J$ of order $N + 1$ with a symmetric spectrum. For ease of notation, we define $\mathcal{A}: \R \to \C$ by 
\begin{equation}\label{eq:A}
\mathcal{A}(t) = 
\begin{cases}
    \displaystyle \sum_{k = 0}^{n} c_k\cos{\lambda_k t} &\text{if $N + 1 = 2n + 1$ is odd} \\
    \displaystyle \sum_{k = 1}^{n} \widetilde{c}_k\cos{\lambda_k t} &\text{if $N + 1 = 2n$ is even}
\end{cases}
\end{equation}
so that $\mathcal{A}(t) = \langle e^{-iJt}\e_0, \e_0 \rangle_{\mathbb{C}^{N+1}} 
$. As a further simplification, we will assume from here on that a symmetric $\sigma(J)$ is \emph{reduced}. That is, its ordered positive eigenvalues $0 < \lambda_1 < \cdots < \lambda_n$ of $J$ satisfy the following conditions:
\begin{itemize}
    \item The eigenvalues are integers that alternate in parity, with $\lambda_1$ being odd and $\gcd(\lambda_1, \dots, \lambda_n) = 1$ whenever $N + 1 = 2n + 1$.
    \item The eigenvalues are odd half-integers with $\gcd(2\lambda_1, \dots, 2\lambda_n) = 1$, and each consecutive difference $\lambda_{k + 1} - \lambda_k$ is odd whenever $N + 1 = 2n$.
\end{itemize}
Indeed, if the positive eigenvalues have a common factor, the only change compared to the reduced spectrum is time scaling (see Figures~\ref{fig:123} and~\ref{fig:369}). Evidently, the first instance of PST in a reduced symmetric spectrum is at $T_0 = \pi$ in either case (see~\cite{Kay10}).

\begin{figure}[ht]
    \centering
    \begin{minipage}{0.45\textwidth}
    \includegraphics[height=1.4in]{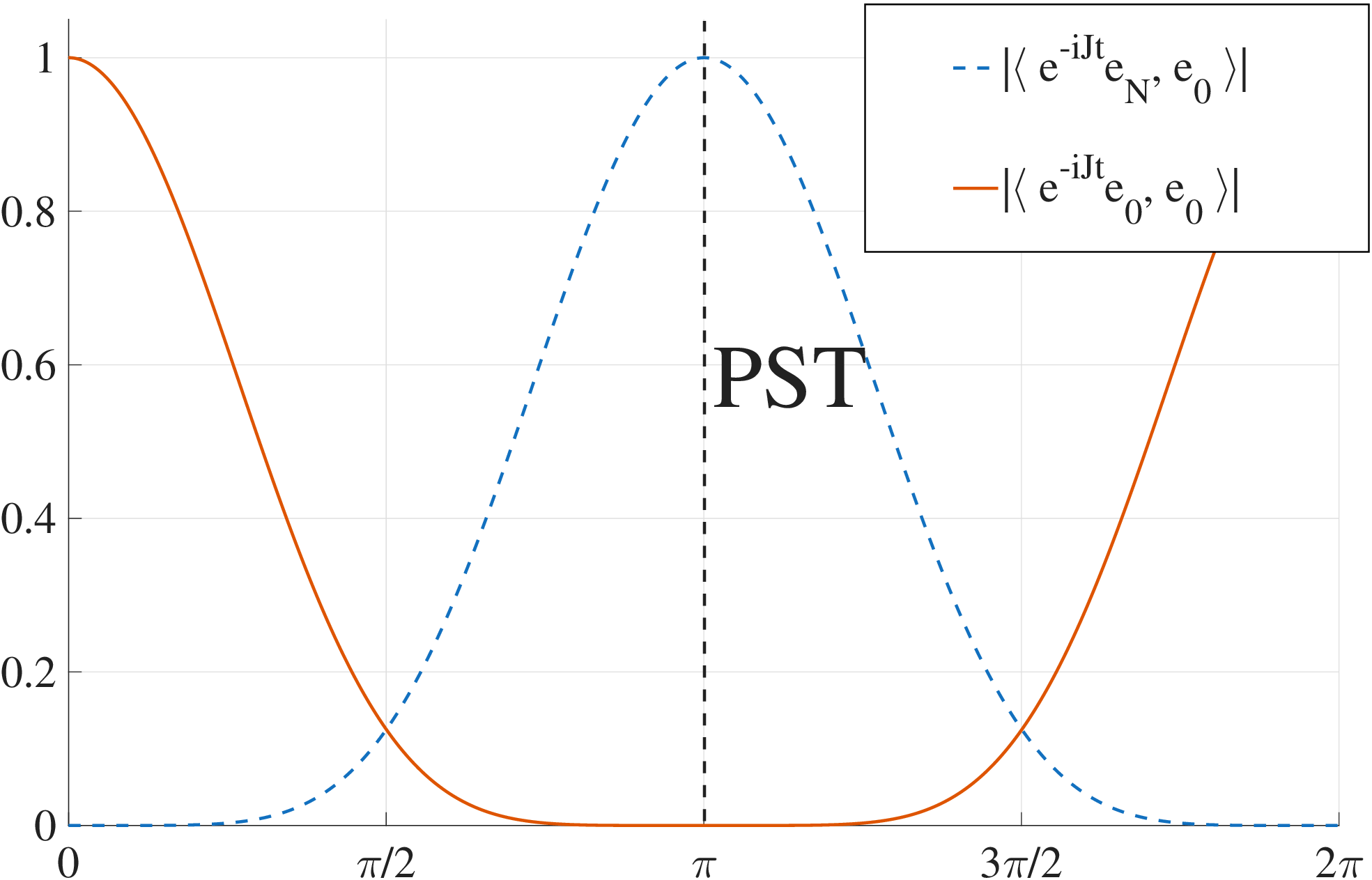}
    \caption{Examples of PST occurrences for $J$ with the reduced equidistant spectrum $\{0,\pm1,\pm2,\pm3\}$.}
    \label{fig:123}
    \end{minipage}\hfill
    \begin{minipage}{0.45\textwidth}
    \centering
    \includegraphics[height=1.4in]{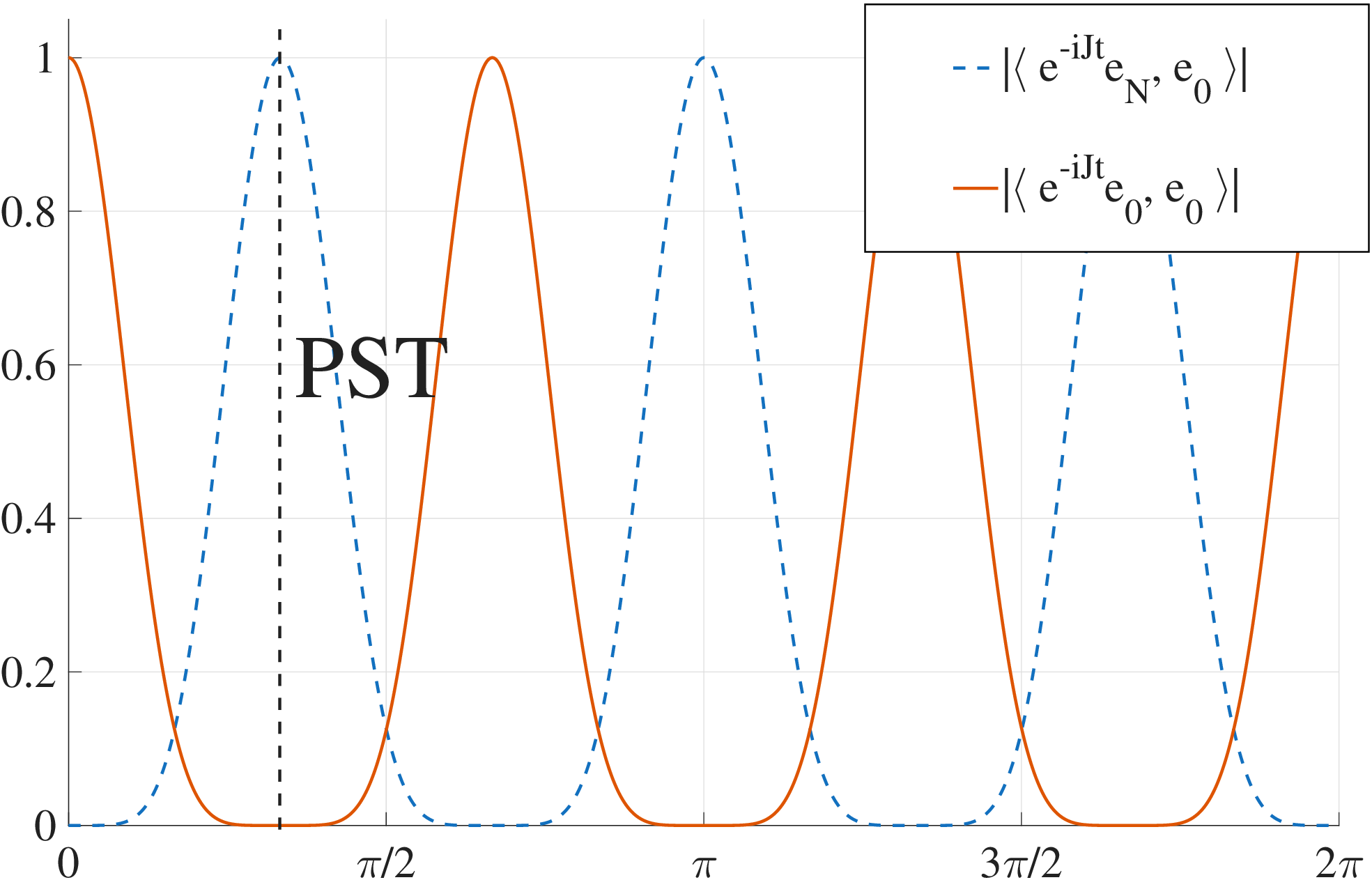}
    \caption{Examples of PST occurrences for $J$ with the equidistant spectrum $\{0,\pm3,\pm6,\pm9\}$.}
    \label{fig:369}
    \end{minipage}
\end{figure}

\begin{lemma}
\label{Lemma 1.4}
Suppose that $J$ is a Jacobi matrix of order $2n + 1$ realizing PST with a reduced symmetric spectrum. 
Then $\mathcal{A}(t)$ has a zero of multiplicity $2n$ at $t = \pi$. Furthermore,
\[
\mathcal{A}^{(2n)}(\pi) = \frac{1}{\displaystyle \sum_{j = 0}^{n} \prod_{\substack{i = 0 \\ i \neq j}}^{n} \frac{1}{|\lambda_j^2 - \lambda_i^2|}}.
\]
\end{lemma}
\begin{proof}
Clearly, any odd-ordered derivative of $\mathcal{A}(t)$ vanishes at $t = \pi$. Thus, it suffices to calculate the value of $\mathcal{A}^{(2m)}(\pi)$ for every integer $0 \leq m \leq n$.

Consider the normalization constant
\begin{equation*}
\label{4.4}
\mathcal{N} = \sum_{j = 0}^n \prod_{\substack{i = 0 \\ i \neq j}}^n \frac{1}{|\lambda_j^2 - \lambda_i^2|}.
\end{equation*}
Note that the coefficients in \eqref{eq:c_k} are exactly
\[
c_k = \frac{1}{\mathcal{N}} \prod_{\substack{i = 0 \\ i \neq k}}^n \frac{1}{|\lambda_k^2 - \lambda_i^2|} = \frac{(-1)^{k + n}}{\mathcal{N}} \prod_{\substack{i = 0 \\ i \neq k}}^n \frac{1}{\lambda_k^2 - \lambda_i^2}.
\]
Therefore
\begin{align*}
\mathcal{A}^{(2m)}(\pi) &= \sum_{k = 0}^{n} (-1)^{k + m}\lambda_k^{2m}c_k \\
    &= \frac{(-1)^{m + n}}{\mathcal{N}} \sum_{k = 0}^{n} \lambda_k^{2m} \prod_{\substack{i = 0 \\ i \neq k}}^n \frac{1}{\lambda_k^2 - \lambda_i^2} \\
    &= \frac{(-1)^{m + n}\mathcal{P}[\lambda^2_0,\ldots,\lambda^2_n]}{\mathcal{N}} \\
    &=\frac{(-1)^{m + n}\delta_{mn}}{\mathcal{N}},
\end{align*}
where $\mathcal{P}[\lambda_0^2,\ldots,\lambda^2_n]$ denotes the divided difference of $\mathcal{P}(x) = x^m$ and $\delta_{ij}$ is the Kronecker delta function.
\end{proof}

\begin{remark}
    Similarly, in the case when $J$ is a Jacobi matrix of order $2n$ realizing PST with a reduced symmetric spectrum, 
 $\mathcal{A}(t)$ has a zero of multiplicity $2n-1$ at $t = \pi$ and
\[
\mathcal{A}^{(2n-1)}(\pi) = \frac{1}{\displaystyle \sum_{j = 1}^{n}\frac{1}{\lambda_j} \prod_{\substack{i = 1 \\ i \neq j}}^{n} \frac{1}{|\lambda_j^2 - \lambda_i^2|}}.
\]
\end{remark}
\begin{figure}[h!]
    \centering
    \includegraphics[width=0.5\linewidth]{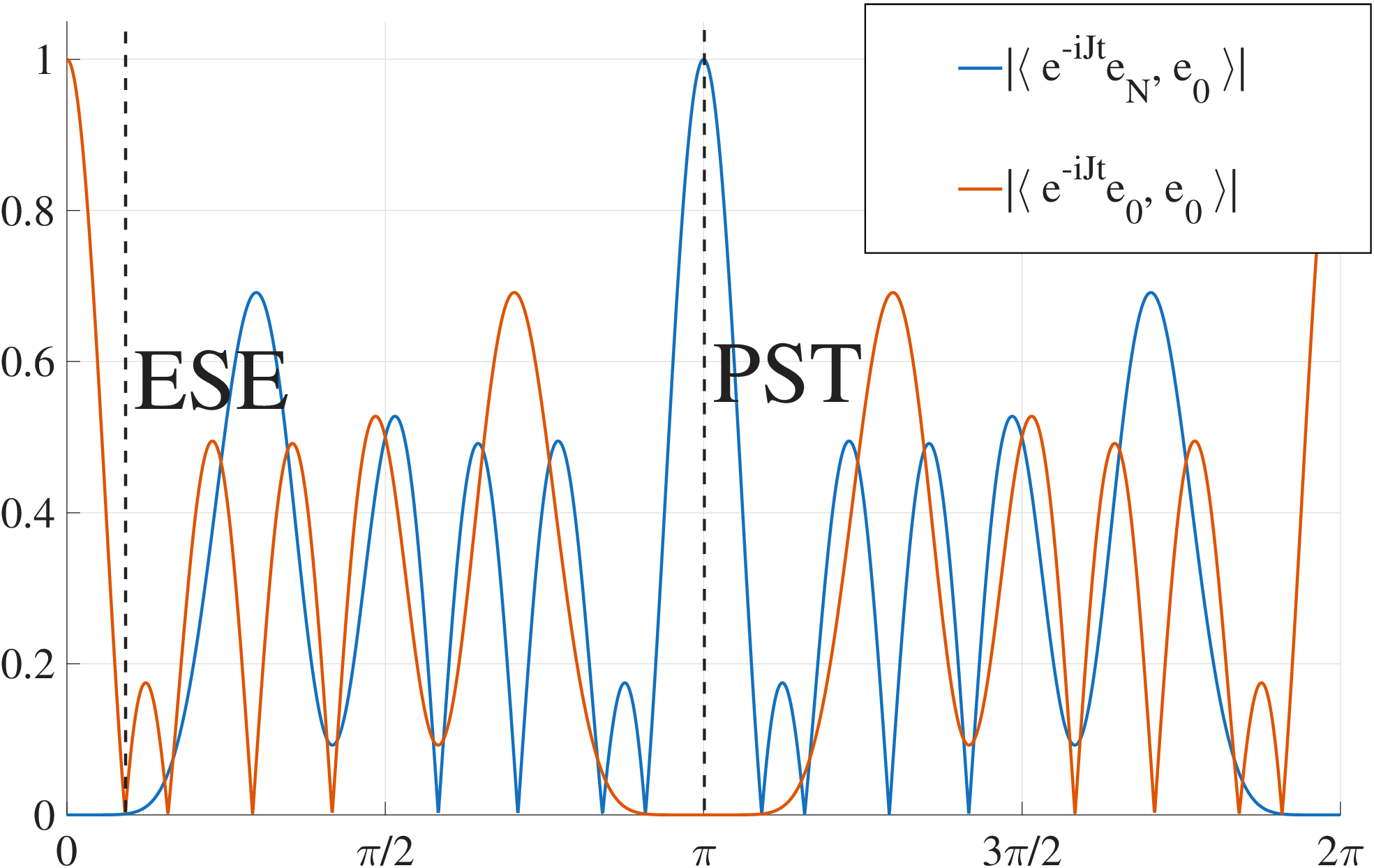}
    \caption{{Example of the ESE occurrences for $J$ with the reduced spectrum $\{0,\pm 3, \pm 8, \pm 9, \pm 12\}$. Given that $J$ realizes PST for the first time at $T_0=\pi$, then $\mathcal{A}(t)$ has a root of multiplicity $10$ at $t = \pi$.}}
    \label{fig:placeholder}
\end{figure}

As a corollary, we obtain an upper bound for the number of times $J$ achieves ESE whenever $J$ is of odd order.

\begin{corollary}\label{cor:odd case number of zeros}
A Jacobi matrix of order $2n + 1$ realizing PST with a reduced symmetric spectrum exhibits ESE an even number of times, up to a maximum of $\lambda_n - n$ instances.
\end{corollary}
\begin{proof}
Let $P(x)$ be the unique real polynomial satisfying $P(\cos(t))=\mathcal{A}(t)$. Since $\mathcal{A}(t)$ has a root of order $2n$ at $t = \pi$ by our previous lemma, $P(x)$  has a root of order of order $n$ at $x = -1$. Hence, $P(x) = (x + 1)^nQ(x)$ for some unique real polynomial $Q(x)$. By the Fundamental Theorem of Algebra, $Q(x)$ has at most $\deg Q = \deg P - n = \lambda_n - n$ roots (counting multiplicity) within the open interval $(-1, 1)$. Because complex roots of $Q(x)$ occur in conjugate pairs and $\lambda_n - n$ is even, $Q(x)$ admits an even number of real zeros in this interval. Thus, by the bijective correspondence between the zeros of $P(x)$ in $(-1, 1)$ and the zeros of $\mathcal{A}(t)$ in $(0, \pi)$, there exist at most $\lambda_n - n$ zeros of $\mathcal{A}(t)$ in $(0, \pi)$. The result then follows from the definition of $\mathcal{A}(t)$ in \eqref{eq:A}.
\end{proof}
\begin{corollary}\label{cor:even case number of zeros}
A Jacobi matrix of order $2n$ realizing PST with a reduced symmetric spectrum exhibits ESE a finite number of times, up to a maximum of $2\lambda_n - n$ instances.
\end{corollary}
In this case, polynomial $P(x)$ can be chosen to satisfy $P(\cos(t/2))=\mathcal{A}(t)$. 

\section{Spectral Fusion}
\label{Section 4}
In contrast to the ``spectral surgery'' framework introduced in~\cite{VZh12}, which removes a part of the spectrum, we introduce \emph{spectral fusion} as its inverse operation. This process involves the addition of an even number of eigenvalues symmetrically at the end of the chain. By doing so, we construct Hamiltonians for a spin chain elongated by two qubits while conserving the desired properties.

\subsection{Why does it work?}
\par Fix a persymmetric Jacobi matrix $J$ of order $N + 1$, and assume that its spectrum $\sigma(J)$ is both reduced and symmetric. Define a sequence of persymmetric Jacobi matrices $\{J_m\}_{m = 0}^\infty$ by asserting that $\sigma(J_m) = \sigma(J) \cup \{\pm \lambda_{n + 1}\}$ (see~\cite[Theorem 3]{H74}), where
\[
\lambda_{n + 1} = \lambda_n + 2m + 1.
\]
Correspondingly, define a sequence of functions $\mathcal{A}_m: \mathbb{R} \to \mathbb{C}$ by
\begin{equation*}
\mathcal{A}_m(t) =
\begin{cases}
    \displaystyle \sum_{k = 0}^{n + 1} d_k \cos{\lambda_kt} &\text{if $N + 1 = 2n + 1$ is odd} \\
    \displaystyle \sum_{k = 1}^{n + 1} \widetilde{d}_k \cos{\lambda_kt} &\text{if $N + 1 = 2n$ is even}
\end{cases}
\end{equation*}
where
\begin{equation*}
d_k = \prod_{\substack{i = 0 \\ i \neq k}}^{n + 1} \frac{1}{|\lambda_k^2 - \lambda_i^2|} \Bigg/\sum_{j = 0}^{n + 1} \prod_{\substack{i = 0 \\ i \neq j}}^{n + 1} \frac{1}{|\lambda_j^2 - \lambda_i^2|} \quad \text{for every $0 \leq k \leq n + 1$}
\end{equation*}
and
\begin{equation*}
\widetilde{d}_k = \frac{1}{\lambda_k} \prod_{\substack{i = 1 \\ i \neq k}}^{n + 1} \frac{1}{|\lambda_k^2 - \lambda_i^2|} \Bigg/\sum_{j = 1}^{n + 1} \frac{1}{\lambda_j}  \prod_{\substack{i = 1 \\ i \neq j}}^{n + 1} \frac{1}{|\lambda_j^2 - \lambda_i^2|} \quad \text{for every $1 \leq k \leq n + 1$}
\end{equation*}
\begin{lemma}
$\mathcal{A}_m \to \mathcal{A}$ uniformly as $m \to \infty$.
\end{lemma}
\begin{proof}
Note that
\[
\mathcal{A}_m(t) - \mathcal{A}(t) =
\begin{cases}
\displaystyle d_{n + 1} \cos{\lambda_{n + 1}t} + \sum_{k = 0}^{n} (d_k - c_k) \cos{{\lambda_kt}}  &\text{if $N + 1 = 2n + 1$ is odd} \\
\displaystyle \widetilde{d}_{n + 1} \cos{\lambda_{n + 1}t} + \sum_{k = 1}^{n} (\widetilde{d}_k - \widetilde{c}_k) \cos{{\lambda_kt}} &\text{if $N + 1 = 2n$ is even}
\end{cases}
\]
for every integer $m \geq 0$. Hence,
\[
\sup_{t \in \mathbb{R}} |\mathcal{A}_m(t) - \mathcal{A}(t)| \leq
\begin{cases}
\displaystyle d_{n + 1} + \sum_{k = 0}^n |d_k - c_k| &\text{if $N + 1 = 2n + 1$ is odd} \\
\displaystyle \widetilde{d}_{n + 1} + \sum_{k = 1}^n |\widetilde{d}_k - \widetilde{c}_k| &\text{if $N + 1 = 2n$ is even.}
\end{cases}
\]
To show that $\mathcal{A}_m \to \mathcal{A}$ uniformly, it suffices to check that
\begin{itemize}
    \item $d_{n + 1} \to 0$ and $d_k \to c_k$ (for every $0 \leq k \leq n$) as $m \to \infty$ whenever $N + 1$ is odd, and
    \item $\widetilde{d}_{n + 1} \to 0$ and $\widetilde{d}_k \to \widetilde{c}_k$ (for every $1 \leq k \leq n$) as $m \to \infty$ whenever $N + 1$ is even.
\end{itemize}
We consider each case separately.

Assume that $N + 1 = 2n + 1$ is odd. Since
\begin{align*}
    d_{n+1} &= \frac{\displaystyle \prod_{i = 0}^n \frac{1}{\lambda_{n + 1}^2 - \lambda_i^2}}{\displaystyle \sum_{j = 0}^n \frac{1}{\lambda_{n + 1}^2 - \lambda_j^2} \cdot \prod_{\substack{i = 0 \\ i \neq k}}^{n} \frac{1}{|\lambda_{j}^2 - \lambda_i^2|} + \prod_{\substack{i = 0 }}^n \frac{1}{\lambda_{n + 1}^2 - \lambda_i^2}}
\end{align*}
and
\begin{align*}
d_{k} - c_k &= \left(\displaystyle \frac{1}{(\lambda_{n+1}^2 - \lambda_k^2)\displaystyle\sum_{j = 0}^{n+1} \prod_{\substack{i = 0 \\ i \neq j}}^{n+1} \frac{1}{|\lambda_j^2 - \lambda_i^2|}} - \frac{1}{\displaystyle\sum_{j = 0}^{n} \prod_{\substack{i = 0 \\ i \neq j}}^n \frac{1}{|\lambda_j^2 - \lambda_i^2|}}\right) \cdot \prod_{\substack{i = 0 \\ i \neq k}}^n \frac{1}{|\lambda_k^2 - \lambda_i^2|}\\
&= \displaystyle \frac{\displaystyle\sum_{j = 0}^{n} \prod_{\substack{i = 0 \\ i \neq j}}^n \frac{1}{|\lambda_j^2 - \lambda_i^2|}-(\lambda_{n+1}^2 - \lambda_k^2)\displaystyle\sum_{j = 0}^{n+1} \prod_{\substack{i = 0 \\ i \neq j}}^{n+1} \frac{1}{|\lambda_j^2 - \lambda_i^2|}}{\left((\lambda_{n+1}^2 - \lambda_k^2)\displaystyle\sum_{j = 0}^{n+1} \prod_{\substack{i = 0 \\ i \neq j}}^{n+1} \frac{1}{|\lambda_j^2 - \lambda_i^2|}\right)\left(\displaystyle\sum_{j = 0}^{n} \prod_{\substack{i = 0 \\ i \neq j}}^n \frac{1}{|\lambda_j^2 - \lambda_i^2|}\right)}\cdot\prod_{\substack{i = 0 \\ i \neq k}}^n \frac{1}{|\lambda_k^2 - \lambda_i^2|}
\\
&= \displaystyle \frac{\displaystyle \sum_{j = 0}^{n} \left(1-\frac{\lambda_{n+1}^2 - \lambda_k^2}{\lambda_{n+1}^2 - \lambda_j^2}\right) \cdot \prod_{\substack{i = 0 \\ i \neq j}}^n \frac{1}{|\lambda_j^2 - \lambda_i^2|}-\prod_{\substack{i = 0 \\ i \neq k}}^{n} \frac{1}{\lambda_{n+1}^2 - \lambda_i^2}}{\left((\lambda_{n+1}^2 - \lambda_k^2)\displaystyle\sum_{j = 0}^{n+1} \prod_{\substack{i = 0 \\ i \neq j}}^{n+1} \frac{1}{|\lambda_j^2 - \lambda_i^2|}\right)\left(\displaystyle\sum_{j = 0}^{n} \prod_{\substack{i = 0 \\ i \neq j}}^n \frac{1}{|\lambda_j^2 - \lambda_i^2|}\right)}\cdot\prod_{\substack{i = 0 \\ i \neq k}}^n \frac{1}{|\lambda_k^2 - \lambda_i^2|}
\end{align*}
for all $0 \leq k \leq n$, then $d_{n+1} = O(m^{-2n})$ and $|d_{k}-c_k| = O(m^{-2})$. Hence
\[
\sup_{t \in \mathbb{R}} |\mathcal{A}_m(t) - \mathcal{A}(t)| \leq d_{n + 1} + \sum_{k = 0}^n |d_k - c_k| = O(1/m^2),
\]
implying that $\mathcal{A}_m \to \mathcal{A}$ uniformly as $m \to \infty$ in the odd case.

Now assume that $N + 1 = 2n$ is even. Similarly, since
\[
\widetilde{d}_{n + 1} = \frac{\displaystyle \frac{1}{\lambda_{n + 1}} \prod_{i = 1}^n \frac{1}{\lambda_{n + 1}^2 - \lambda_i^2}}{\displaystyle \sum_{j = 1}^n \frac{1}{\lambda_j(\lambda_{n + 1}^2 - \lambda_j^2)}  \prod_{\substack{i = 1 \\ i \neq j}}^{n} \frac{1}{|\lambda_j^2 - \lambda_i^2|} + \frac{1}{\lambda_{n + 1}}  \prod_{i = 1}^n \frac{1}{\lambda_{n + 1}^2 - \lambda_i^2}}
\]
and
\begin{align*}
\widetilde{d}_k - \widetilde{c}_k &= \left(\displaystyle \frac{1}{(\lambda_{n+1}^2 - \lambda_k^2)\displaystyle \sum_{j = 1}^{n+1}  \frac{1}{\lambda_j} \prod_{\substack{i = 1 \\ i \neq j}}^{n+1} \frac{1}{|\lambda_j^2 - \lambda_i^2|}} - \frac{1}{\displaystyle \sum_{j = 1}^{n}  \frac{1}{\lambda_j} \prod_{\substack{i = 1 \\ i \neq j}}^n \frac{1}{|\lambda_j^2 - \lambda_i^2|}}\right) \cdot \frac{1}{\lambda_k} \prod_{\substack{i = 1 \\ i \neq k}}^n \frac{1}{|\lambda_k^2 - \lambda_i^2|} \\
    &= \frac{\displaystyle \sum_{j = 1}^{n}  \frac{1}{\lambda_j} \prod_{\substack{i = 1 \\ i \neq j}}^n \frac{1}{|\lambda_j^2 - \lambda_i^2|} - (\lambda_{n+1}^2 - \lambda_k^2) \sum_{j = 1}^{n+1}  \frac{1}{\lambda_j} \prod_{\substack{i = 1 \\ i \neq j}}^{n+1} \frac{1}{|\lambda_j^2 - \lambda_i^2|}}{\displaystyle \left((\lambda_{n+1}^2 - \lambda_k^2) \sum_{j = 1}^{n+1}  \frac{1}{\lambda_j} \prod_{\substack{i = 1 \\ i \neq j}}^{n+1} \frac{1}{|\lambda_j^2 - \lambda_i^2|}\right) \left(\sum_{j = 1}^{n}  \frac{1}{\lambda_j} \prod_{\substack{i = 1 \\ i \neq j}}^n \frac{1}{|\lambda_j^2 - \lambda_i^2|}\right)} \cdot \frac{1}{\lambda_k} \prod_{\substack{i = 1 \\ i \neq k}}^n \frac{1}{|\lambda_k^2 - \lambda_i^2|} \\
        &= \frac{\displaystyle \sum_{j = 1}^n \left(1 - \frac{\lambda_{n + 1}^2 - \lambda_k^2}{\lambda_{n + 1}^2 - \lambda_j^2}\right) \cdot \prod_{\substack{i = 1 \\ i \neq j}}^n \frac{1}{|\lambda_j^2 - \lambda_i^2|} - \frac{1}{\lambda_{n + 1}} \prod_{\substack{i = 1 \\ i \neq k}}^n \frac{1}{\lambda_{n + 1}^2 - \lambda_i^2}}{\displaystyle \left((\lambda_{n+1}^2 - \lambda_k^2) \sum_{j = 1}^{n+1}  \frac{1}{\lambda_j} \prod_{\substack{i = 1 \\ i \neq j}}^{n+1} \frac{1}{|\lambda_j^2 - \lambda_i^2|}\right) \left(\sum_{j = 1}^{n}  \frac{1}{\lambda_j} \prod_{\substack{i = 1 \\ i \neq j}}^n \frac{1}{|\lambda_j^2 - \lambda_i^2|}\right)} \cdot \frac{1}{\lambda_k} \prod_{\substack{i = 1 \\ i \neq k}}^n \frac{1}{|\lambda_k^2 - \lambda_i^2|}
\end{align*}
for all $1 \leq k \leq n$, then $\widetilde{d}_{n + 1} = O(m^{-(2n - 1)})$ and $|\widetilde{d}_k - \widetilde{c}_k| = O(m^{-1})$. We similarly arrive at
\[
\sup_{t \in \R} |\mathcal{A}_m(t) - \mathcal{A}(t)| \leq \widetilde{d}_{n + 1} + \sum_{k = 1}^n |\widetilde{d}_k - \widetilde{c}_k| = O(1/m),
\]
implying that $\mathcal{A}_m \to \mathcal{A}$ uniformly as $m \to \infty$ in the even case as well.
\end{proof}

\begin{remark}
    Similar argument shows that $\mathcal{A}'_m $ converges to $\mathcal{A}'$ uniformly as $m \to \infty$.
\end{remark}

\subsection{Main Result}
The next theorem demonstrates a method for constructing infinite families of Jacobi matrices that exhibit PST, with or without ESE, from lower-order PST matrices exhibiting the same ESE property.
\begin{figure}[H]
    \centering
        \begin{minipage}{0.45\textwidth}
           \centering
    \includegraphics[height=1.4in]{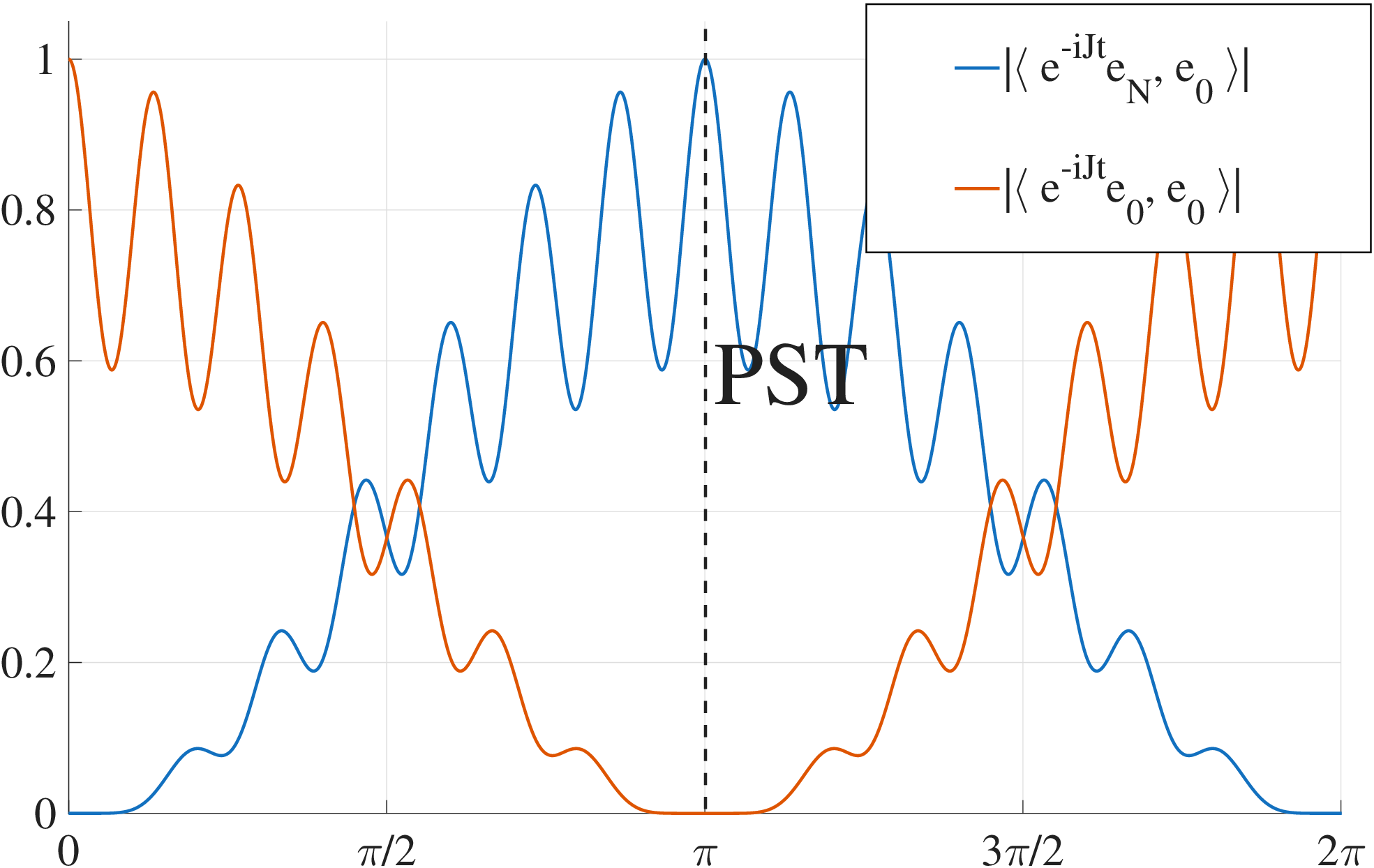}
    \caption{Example of $J$ with reduced spectrum $\{ 0,\pm1,\pm14,\pm15,\pm16\}$.}
    \label{fig: not true}
    \end{minipage}
\end{figure}
\begin{figure}[H]
    \centering
    \begin{minipage}{0.45\textwidth}
    \centering
    \includegraphics[height=1.4in]{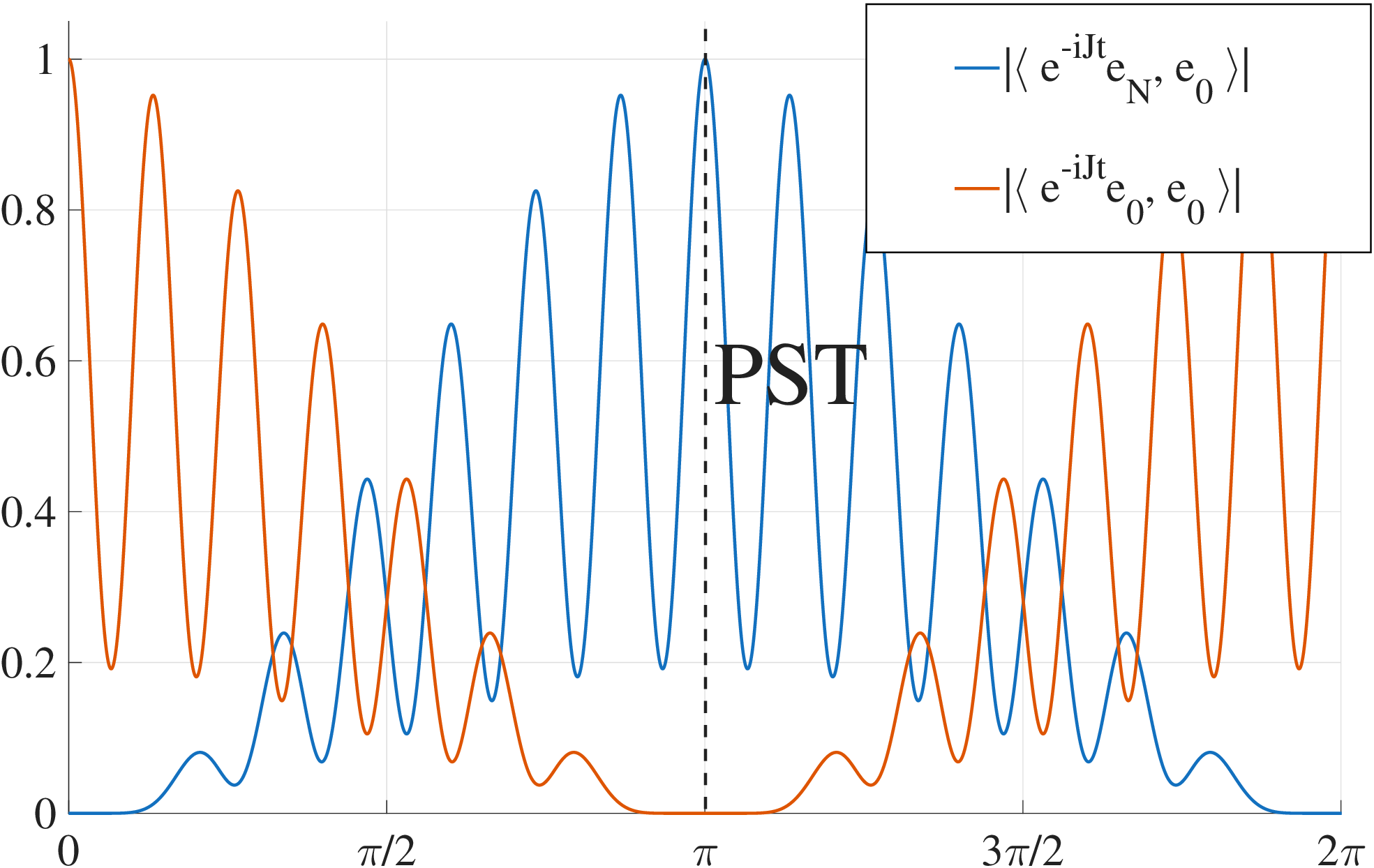}
    \caption{Example of no ESE occurrences for $J$ with the reduced spectrum $\{0,\pm1,\pm14,\pm15,\pm16,\pm19\}$.}
    \label{fig:1 14-16 19}
    \end{minipage}\hfill
    \begin{minipage}{0.45\textwidth}
    \centering
    \includegraphics[height=1.4in]{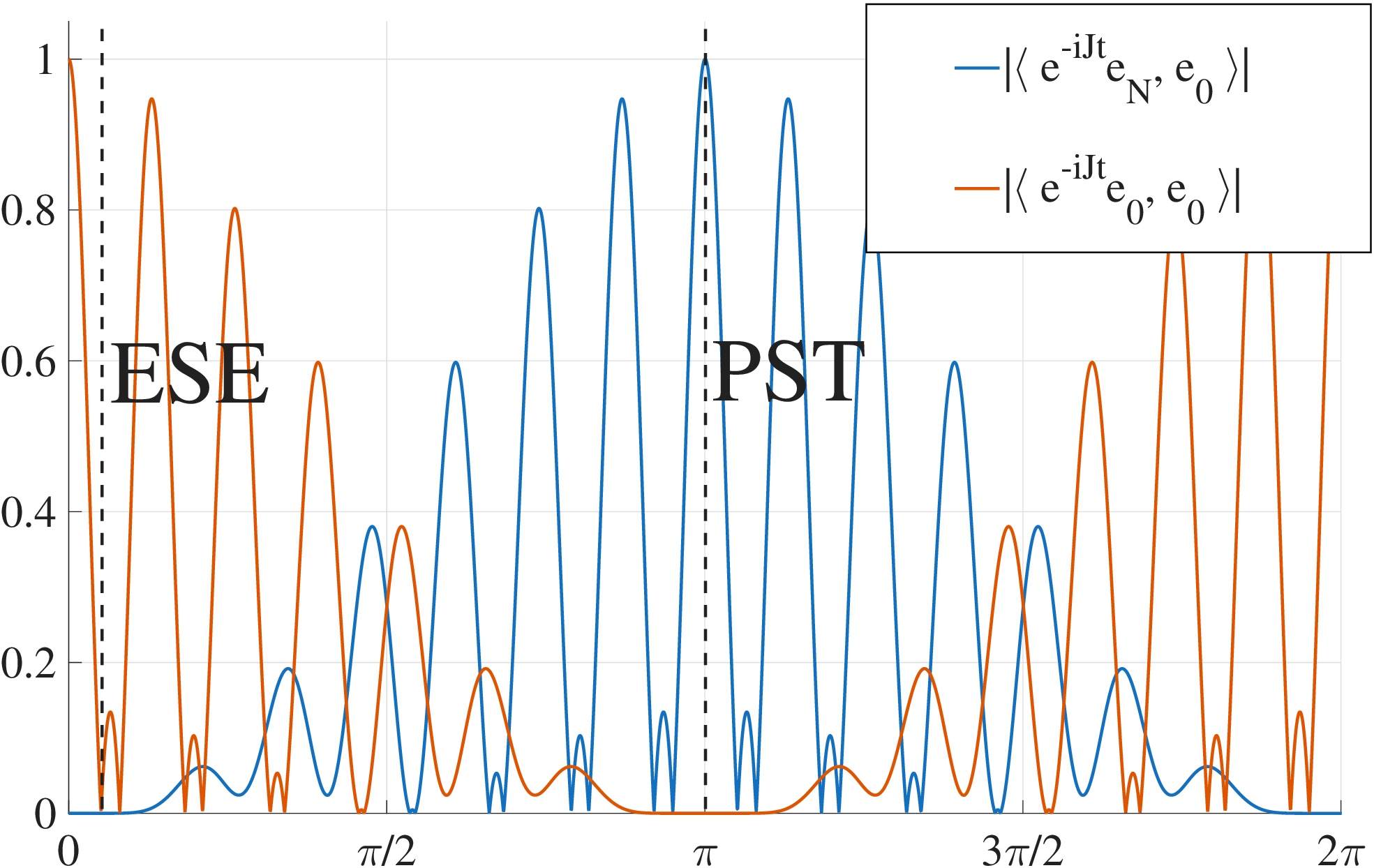}
    \caption{Example of the ESE occurrences for $J$ with the reduced spectrum $\{0,\pm1,\pm14,\pm15,\pm16,\pm17\}$.}
    \label{fig:1 14-17}
    \end{minipage}
\end{figure}

\begin{theorem}
\label{thm: same number of zeros }
Suppose that $J$ is a Jacobi matrix realizing PST and $\mathcal{A}(t)$ has only simple zeros on $(0,\pi)$. Then there exists a positive integer $M$ such that for every integer $m \geq M$, the number of ESE occurrences is identical for the Jacobi matrices $J$ and $J_m$.
\end{theorem}
\begin{proof}
Without loss of generality, assume that $\sigma(J)$ is reduced. From Corollaries~\ref{cor:odd case number of zeros} and \ref{cor:even case number of zeros} we know that $\mathcal{A}$ has finitely many zeros in $(0, \pi)$. Since $\{\mathcal{A}_m\}_{m = 1}^\infty$ uniformly converges to $\mathcal{A}$, $\mathcal{A}_m $ has at least as many zeros as $\mathcal{A}$ on $(0, \pi)$. Moreover, as the zeros of $\mathcal{A}$ are simple, then there is a positive integer $M$ such that for all $m > \max(\lambda_n, M)$, the functions $\mathcal{A}_m$ and $\mathcal{A}$ have the same number of zeros in $(0, \pi)$ as the derivatives would be non-zero around zeros of $\mathcal{A}$.
\end{proof}

Figures~\ref{fig: not true}, \ref{fig:1 14-16 19}, and~\ref{fig:1 14-17} illustrate the effect of adding eigenvalues to a system without ESE based on their distance from existing clusters. In this case, the baseline cluster would be $\{\pm 14, \pm15, \pm16\}$. While appending $\pm 19$ does not result in ESE, introducing $\pm 17$ as the next pair of consecutive integers successfully yields the phenomenon.



\par It is known that a Jacobi matrix with a reduced symmetric spectrum $\{0, \pm \lambda_1, \pm \lambda_2\}$ has ESE if and only if $\lambda_1 \neq 1$~\cite{EM25}. This is not true in the general case (see Figures~\ref{fig: not true}-\ref{fig:1 14-17}). However, recursively applying Theorem~\ref{thm: same number of zeros } to such matrices generates infinitely many infinite families of Jacobi matrices with PST having ESE and no ESE.
\begin{figure}[H]
    \centering
    \begin{minipage}{0.45\textwidth}
    \centering
    \includegraphics[height=1.4in]{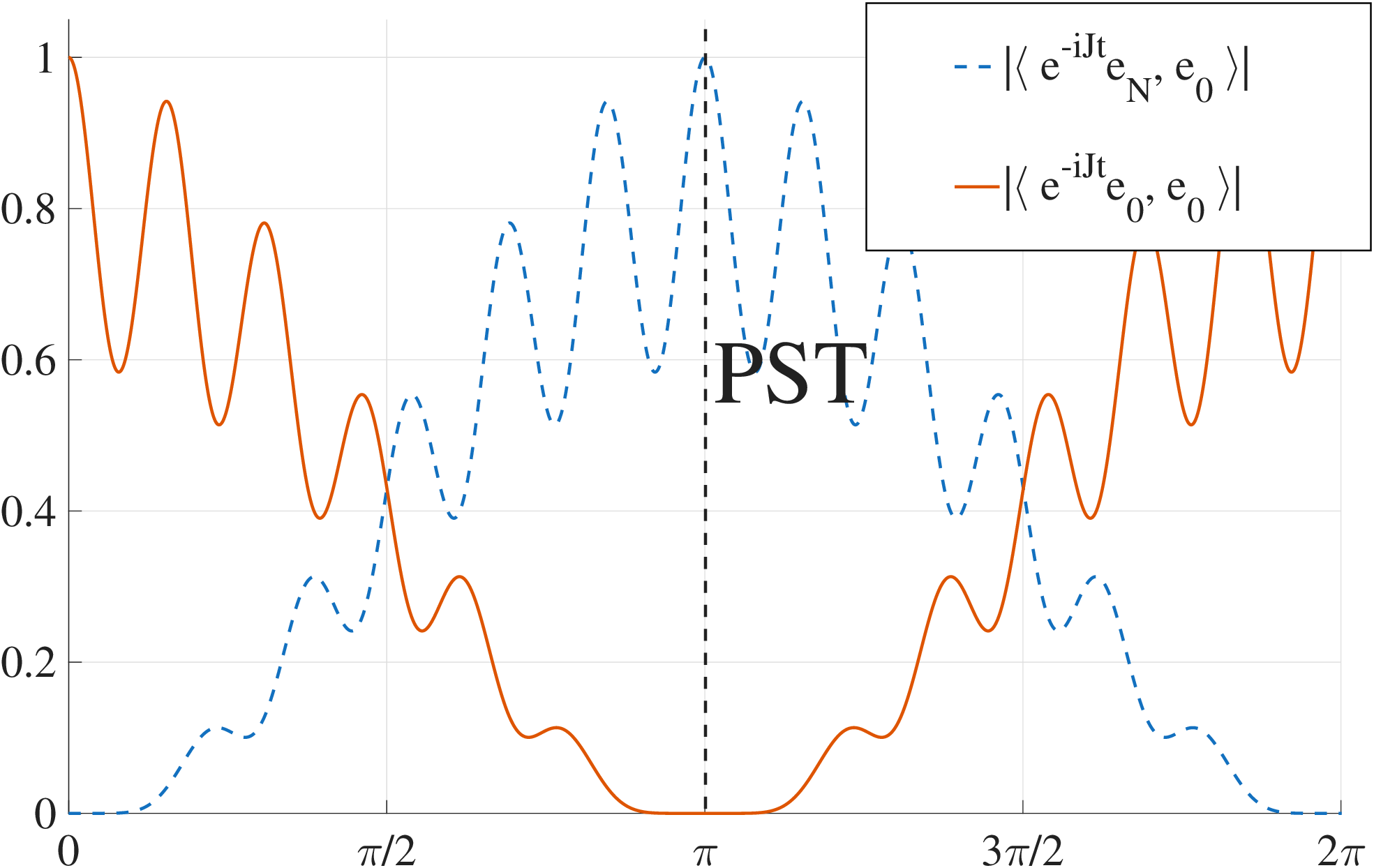}
    \textit{\small $J$ with the reduced spectrum $\{0,\pm1,\pm12,\pm13,\pm14\}$.}
    \end{minipage}\hfill
    \begin{minipage}{0.45\textwidth}
    \centering
    \includegraphics[height=1.4in]{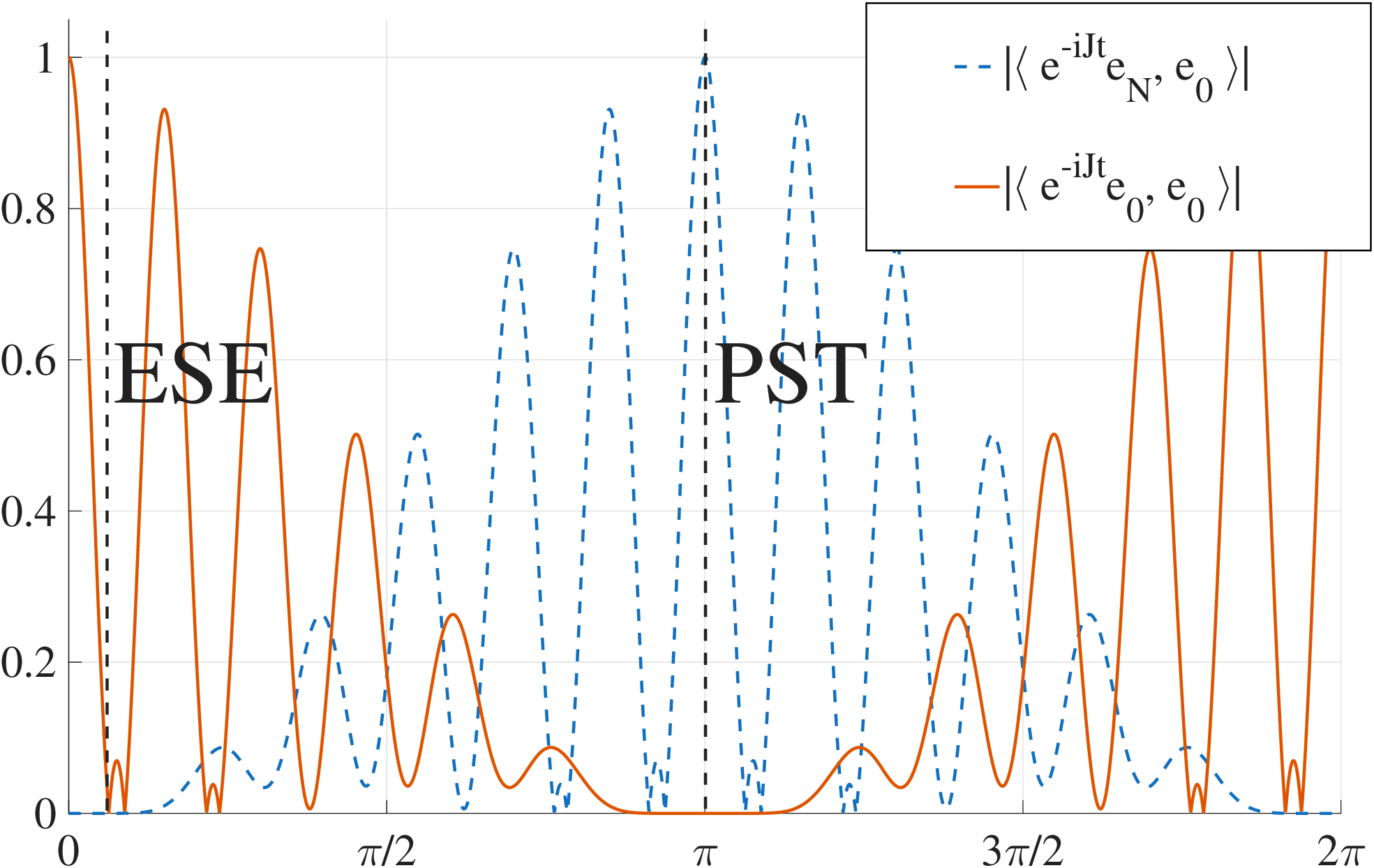}
    \textit{\small$J$ with the reduced spectrum $\{0,\pm1,\pm12,\pm13,\pm14,\pm15\}$.}
    \end{minipage}\hfill
    \begin{minipage}{0.45\textwidth}
    \centering
    \includegraphics[height=1.4in]{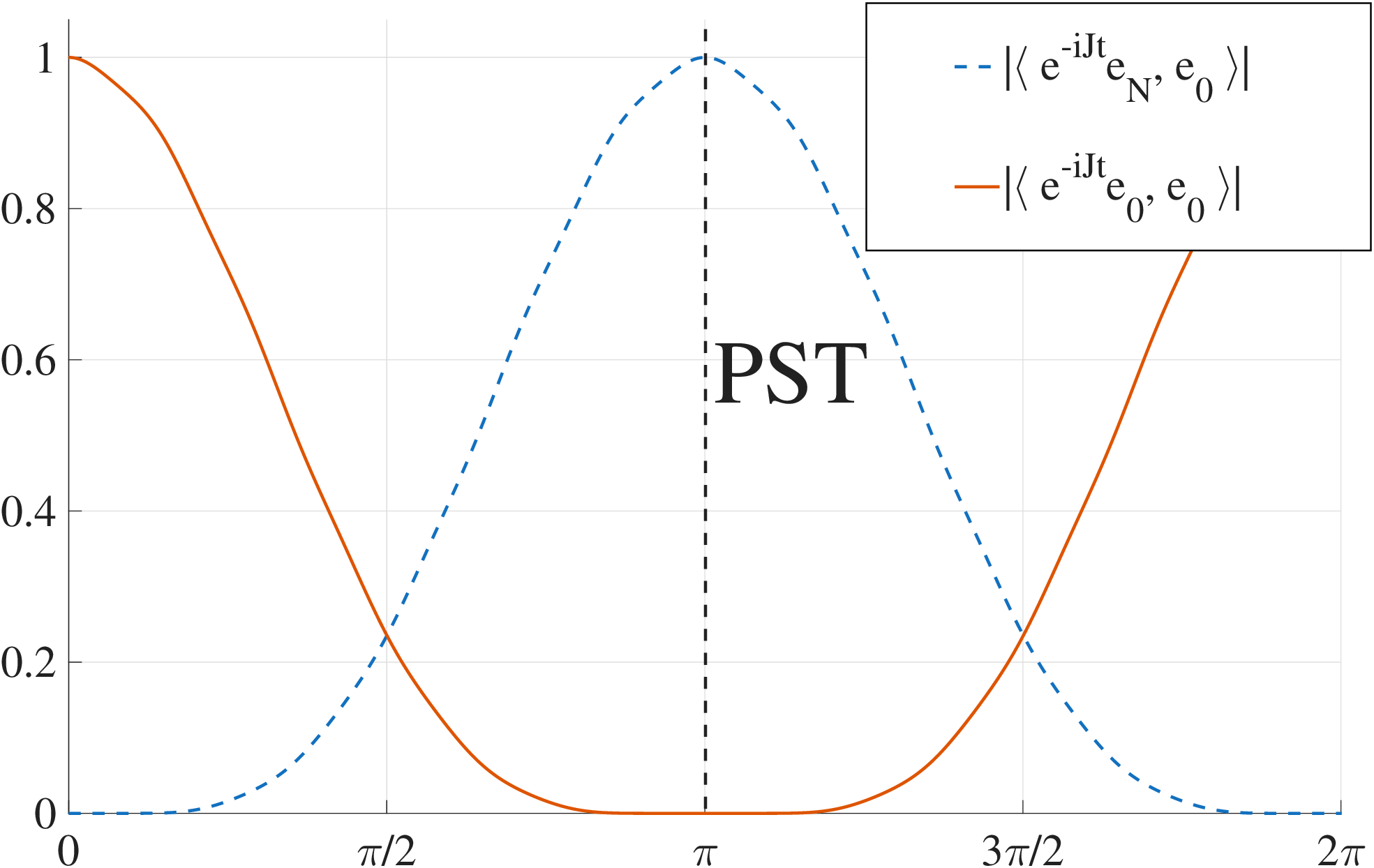}
    \textit{$J$ with the reduced spectrum $\{0,\pm1,\pm2,\pm13,\pm14,\pm15\}$.}
    \end{minipage}
    \hfill
    \begin{minipage}{0.45\textwidth}
    \centering
    \includegraphics[height=1.4in]{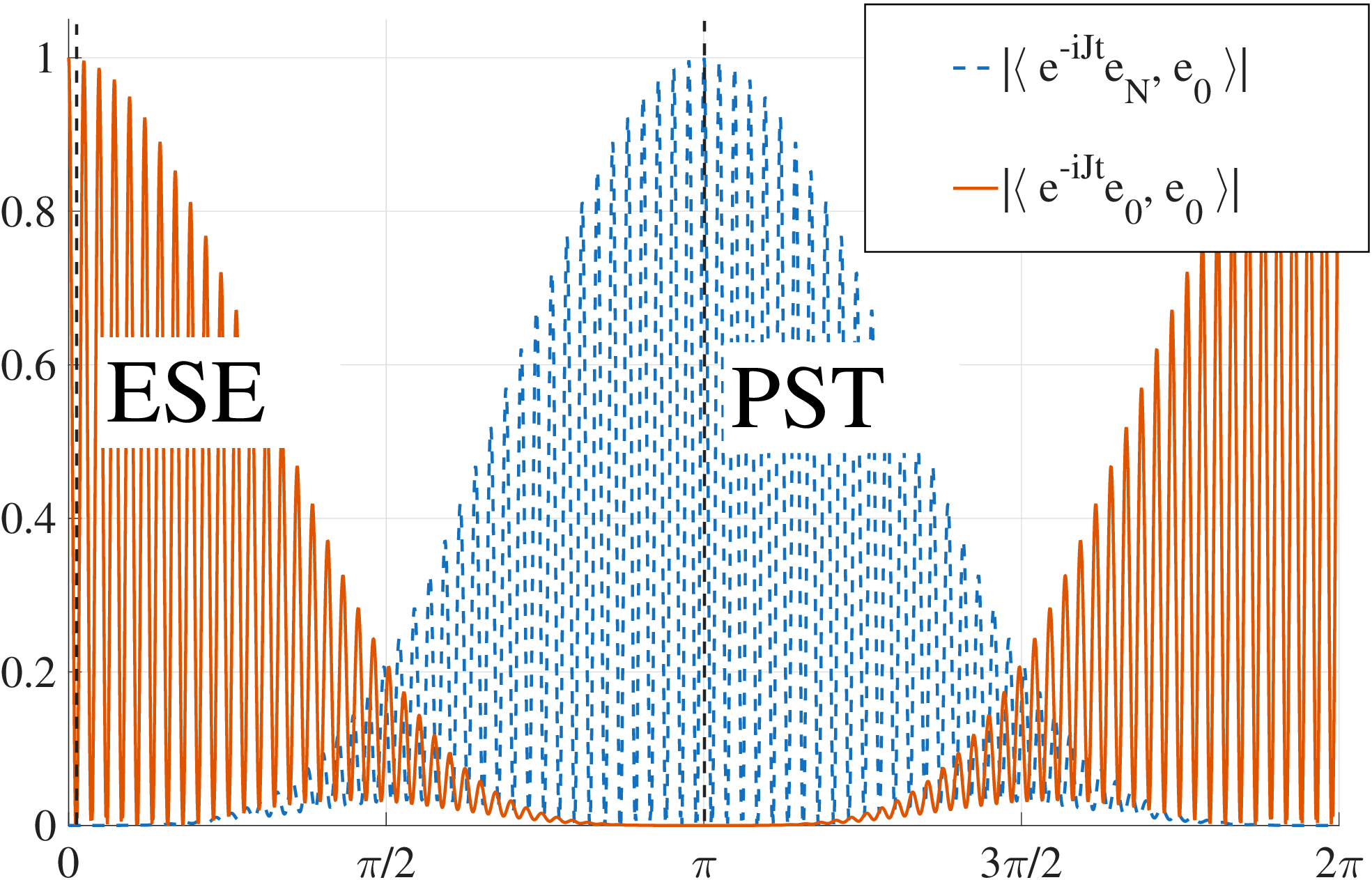}
    \textit{$J$ with the reduced spectrum $\{0,\pm1,\pm2,\pm81,\pm82,\pm83,\pm84,\pm85,\pm86\}$.} 
    \end{minipage} \\
    \begin{minipage}{0.45\textwidth}
    \centering
    \includegraphics[height=1.4in]{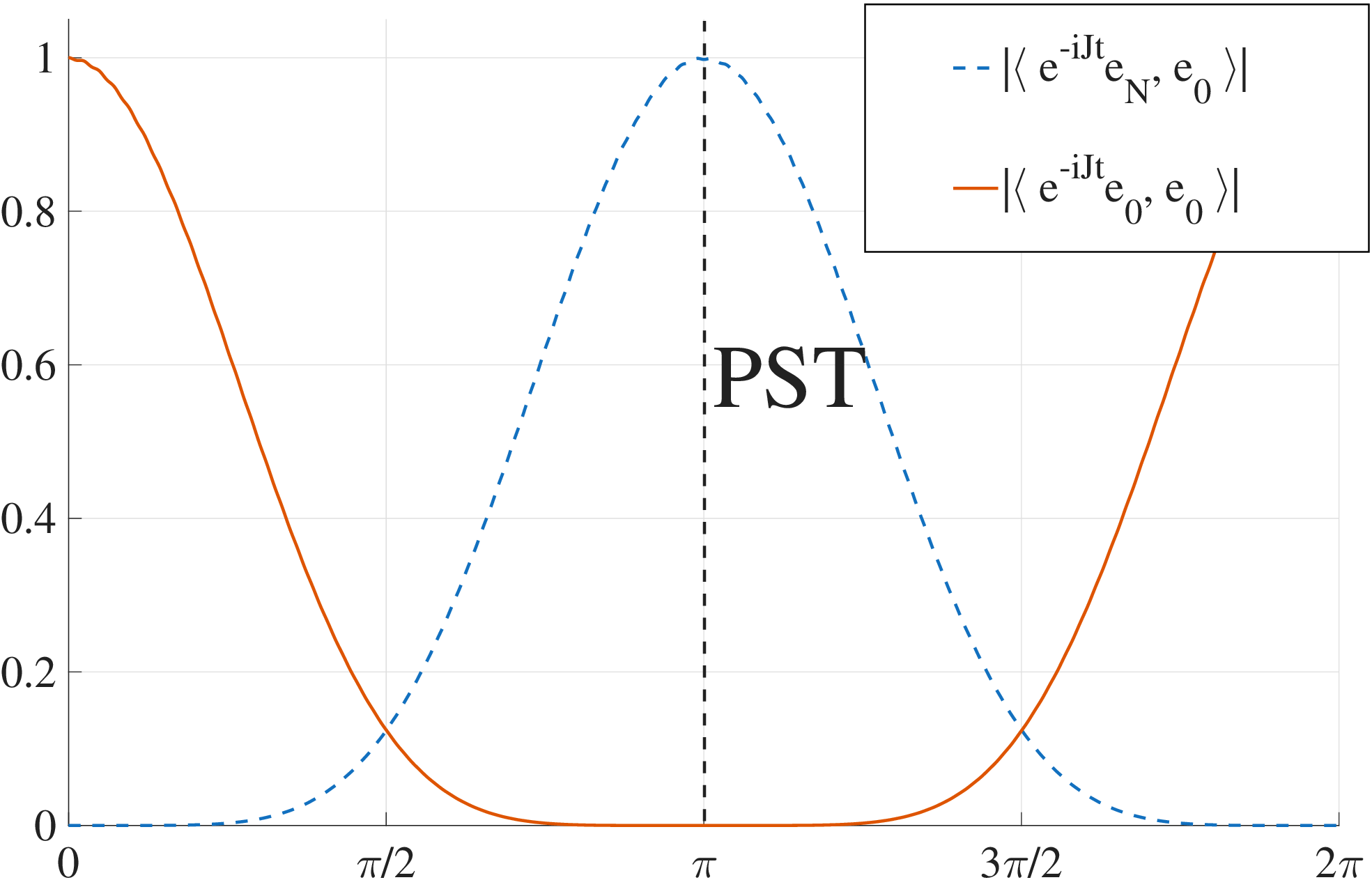}
    \textit{$J$ with the reduced spectrum $\{0,\pm1,\pm2,\pm3,\pm82,\pm83,\pm84,\pm85,\pm86\}$.}
    \end{minipage}
\caption{Examples of having clusters of eigenvalues near $0$ versus clusters further away from $0$.}
\end{figure}

\subsection{The T-Rex Example}
We now discuss the T-Rex chains of arbitrary length $N+1$ (see Figure~\ref{fig:trex}).
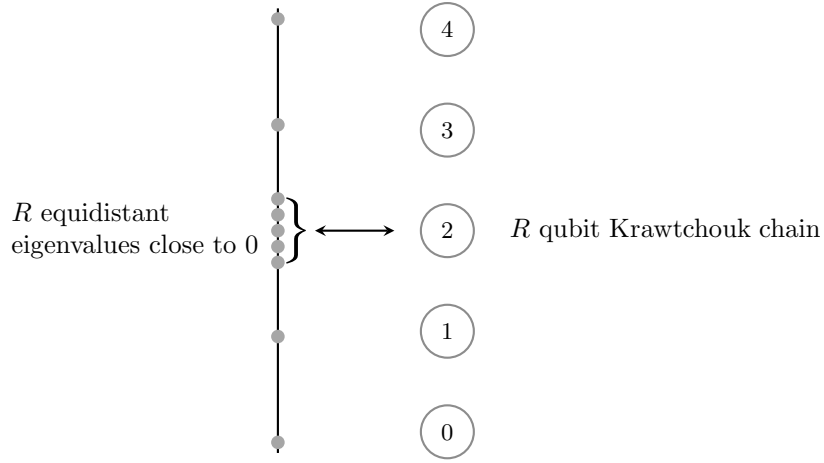
\begin{figure}[ht]
    \centering
\begin{tikzpicture}[
    line/.style={thick, black},
    dot/.style={circle, fill=gray!70, inner sep=1.8pt},
    bigcircle/.style={circle, draw=gray!90, thick, minimum size=0.7cm},
    >=stealth,scale=0.7
]

\draw[line] (0,-4.2) -- (0,4.2);

\foreach \y in {4,2,-2,-4} {
    \node[dot] at (0,\y) {};
}

\foreach \y in {-0.6,-0.3,0,0.3,0.6} {
    \node[dot] at (0,\y) {};
}

\node[align=left] at (-2.7,0) {
    $R$ equidistant\\
    eigenvalues close to $0$
};

\node[align=left] at (0.5,0) {\huge
    $\big\}$
};

\draw[line,<->] (0.7,0) -- (2.2,0);

\foreach \y in {0,1,2,3,4} {
    \node[bigcircle] at (3.2,1.9*\y-3.8) {\small$\y$};
}

\node[anchor=west] at (4.2,0) {
    $R$ qubit Krawtchouk chain
};

\end{tikzpicture}
    \caption{A T-Rex chain of length $N+1$, with a core of $R$ eigenvalues and the rest being large and with large gaps, exhibits no ESE, as does the length-$R$ Krawtchouk chain.}
    \label{fig:trex}
\end{figure}
This construction was originally introduced in~\cite{KKT} to model arrival/departure characteristics. We select a symmetric set of $R$ uniform space eigenvalues with a unit gap (e.g., $\{0,\pm1,\pm2, \dots\}$), where $R$ and $N+1$ share the same parity. The remaining $N-R$ eigenvalues are then chosen symmetrically about the origin to satisfy perfect state transfer conditions at a significantly higher scale $O(m)$, with a mutual separation of $O(m)$. 
\begin{figure}[ht]
\caption{Krawtchouk vs. T-Rex chains}
    \centering
    \begin{minipage}{0.45\textwidth}
    \includegraphics[height=1.4in]{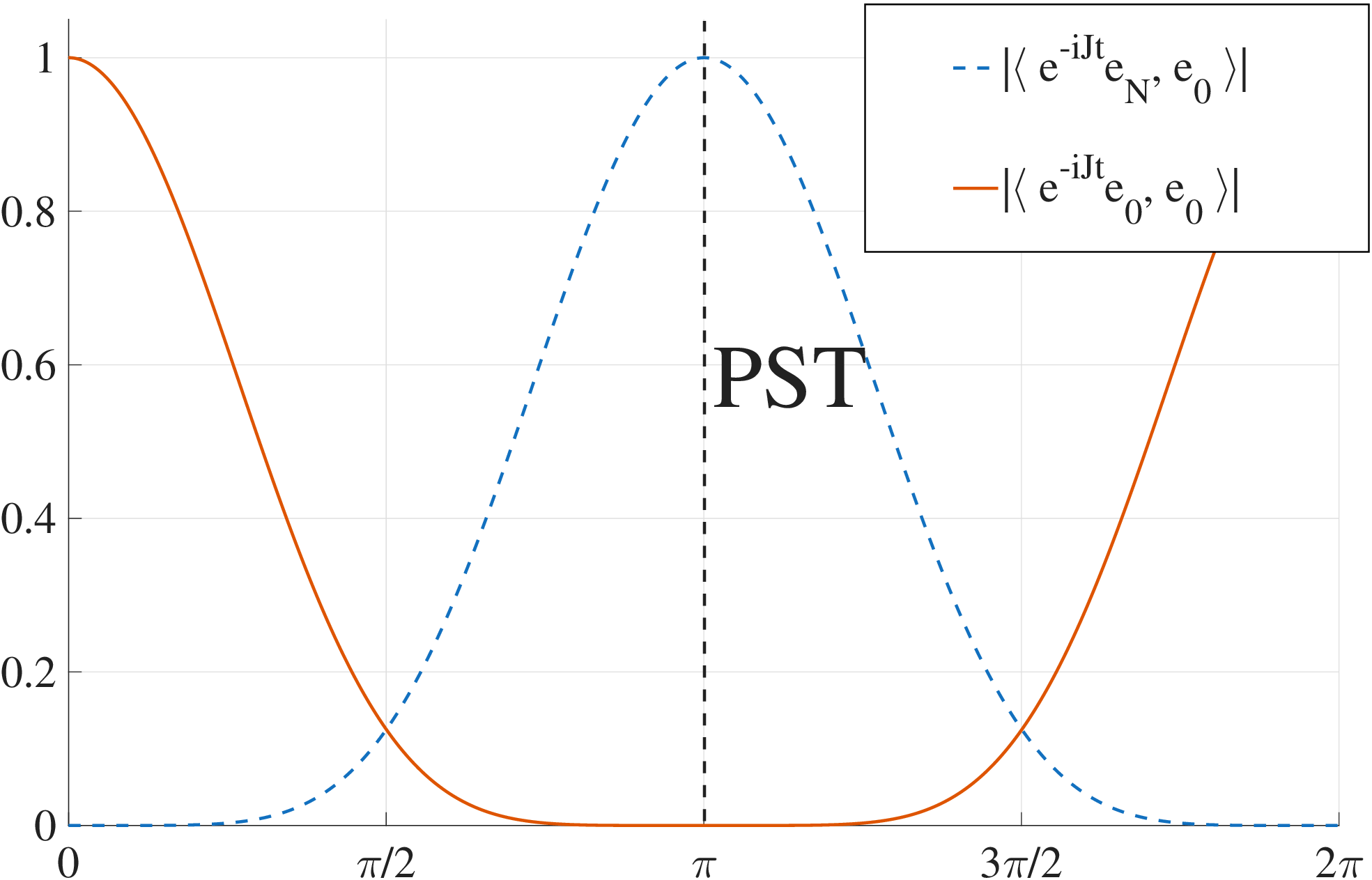}
    \textit{$J$ with the reduced spectrum $\{0,\pm1,\pm2,\pm3,\pm4\}$.}
    \end{minipage}\hfill
    \begin{minipage}{0.45\textwidth}
    \centering
    \includegraphics[height=1.4in]{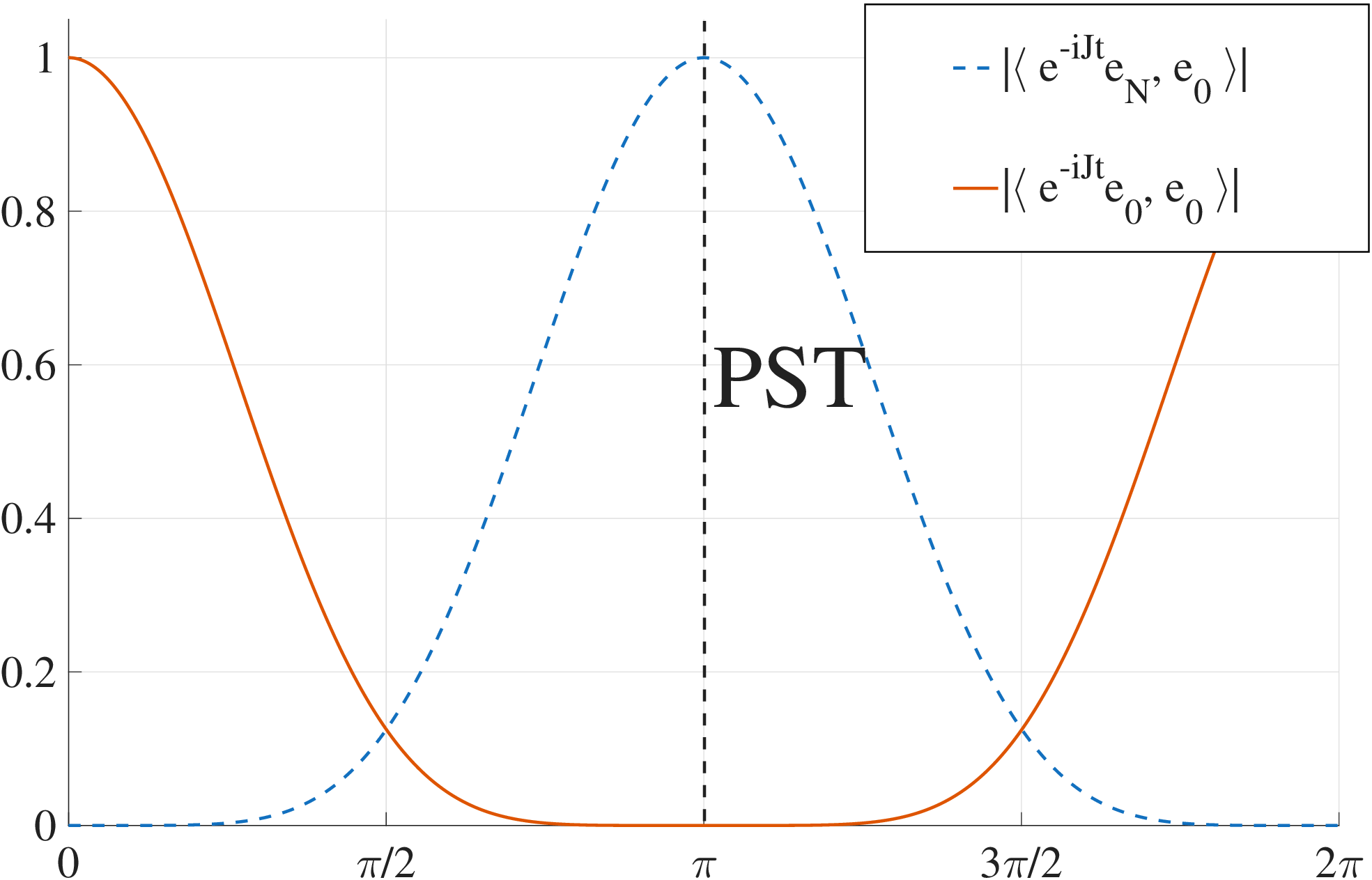}
    \textit{$J$ with the reduced spectrum $\{0,\pm1,\pm2,\pm3,\pm20,\pm25\}$.}
    \end{minipage}
\end{figure}

In~\cite{KKT}, the authors showed that such chains behave essentially like the Krawtchouk chain of length $R$; i.e., with symmetric equidistant spectrum. It was shown in~\cite{ESE} that such Krawtchouk chains do not exhibit ESE. Combining this with our main result, this translates into the following fact.

\begin{corollary}
    For sufficiently large $R$ and $m$, the corresponding T-Rex chain does not have ESE.
\end{corollary}

\section*{Acknowledgements}

This work was supported by the SIAM-Simons Undergraduate Summer Research Program which is funded by the Society for Industrial and Applied Mathematics (SIAM) through award 1036702 of the Simons Foundation. A.M. was also partially supported by the Coffin Grant from the University of Hartford. 

\appendix 
\section{Equidistant Spectra}
Using the amplitude formulas from Proposition~\ref{Prop: amplitudes}, we present an alternative proof for Theorem 3.1 in~\cite{ESE}:

\begin{proposition}{\cite[Theorem 3.1]{ESE}}
Let $J$ be a persymmetric Jacobi matrix of order $N + 1$. Suppose that $J$ has an equidistant spectrum centered about $\mu \in \R$ separated by $\lambda > 0$. Then
\[
\langle e^{-iJt}\e_0, \e_0\rangle_{\mathbb{C}^{N+1}} = e^{-i\mu t} \cos^N\frac{\lambda t}{2}.
\]
Consequently, $J$ does not have ESE.
\end{proposition}
\begin{proof}
As usual, we consider two cases. Assume that $N + 1 = 2n + 1$ is odd. Then the coefficients in \eqref{eq:c_k} are given by
\[
c_0 = \frac{1}{2^{2n}} \binom{2n}{n} \quad \text{and} \quad c_k = \frac{1}{2^{2n - 1}} \binom{2n}{n - k} \quad \text{for every $1 \leq k \leq n$.}
\]
Hence
\begin{align*}
\langle e^{-iJt}\e_0, \e_0\rangle_{\mathbb{C}^{N+1}} &= \frac{e^{-i\mu t}}{2^{2n}} \binom{2n}{n} + \frac{e^{-i\mu t}}{2^{2n - 1}} \sum_{k = 1}^n \binom{2n}{n - k} \cos{k\lambda t} \\
    &= \frac{e^{-i\mu t}}{2^{2n}} \sum_{k = 0}^{2n} \binom{2n}{k} e^{i\lambda(n - k)t} \\
    &= e^{-i\mu t} \cos^{2n}\frac{\lambda t}{2}.
\end{align*}
Now assume that $N + 1 = 2n$ is even. Then the coefficients in \eqref{eq:widetilde{c}_k} are given by 
\[
\widetilde{c}_k = \frac{1}{2^{2n - 2}} \binom{2n - 1}{n + k - 1} \quad \text{for every $1 \leq k \leq n$}
\]
and hence
\begin{align*}
\langle e^{-iJt}\e_0, \e_0\rangle_{\mathbb{C}^{N+1}} &= \frac{e^{-i\mu t}}{2^{2n - 2}}\sum_{k = 1}^n  \binom{2n - 1}{n + k - 1}\cos\frac{(2k - 1)\lambda t}{2} \\
    &= \frac{e^{-i\mu t}}{2^{2n - 1}} \sum_{k = 0}^{2n - 1} \binom{2n - 1}{k}e^{i\lambda (2n + 1 - 2k)t/2} \\
    &= e^{-i\mu t}\cos^{2n - 1}\frac{\lambda t}{2}.
\end{align*}
As the first instance of PST occurs at $T_0 = \pi/\lambda$ for both cases and $\mathcal{A}$ has no roots in $(0, \pi/\lambda)$, then $J$ cannot have ESE. 
\end{proof}

\bibliography{Finaldraft7x7.bib}

\end{document}